\documentclass{amsart}
\usepackage{amsfonts,amsmath,graphicx}
\usepackage{pdfsync}
\usepackage{subfigure}
\usepackage{color}
\newtheorem{theorem}{Theorem}
\theoremstyle{plain}

\def\be{\begin{equation}}
\def\ee{\end{equation}}

\newtheorem{problem}{Problem}

\numberwithin{equation}{section}
\renewcommand\Re{\operatorname{Re}}

\def\Bi{\text{Bi}}
\def\Pr{\text{Pr}}

\begin{document}
\title{Dynamic Transitions of Surface Tension Driven Convection}
\author[Dijkstra]{Henk Dijkstra}
\address[HD]{Institute for Marine and Atmospheric research Utrecht
Department of Physics and Astronomy
Utrecht University
Princetonplein 5, 3584 CC Utrecht, 
The Netherlands}
\email{H.A.Dijkstra@uu.nl}

\author[Sengul]{Taylan Sengul}
\address[TS]{Department of Mathematics,
Indiana University, Bloomington, IN 47405}
\email{msengul@indiana.edu}

\author[Wang]{Shouhong Wang}
\address[SW]{Department of Mathematics,
Indiana University, Bloomington, IN 47405}
\email{showang@indiana.edu, http://www.indiana.edu/~fluid}

\thanks{The work of TS and SW was supported in part by the
Office of Naval Research and by the National Science Foundation.}


\begin{abstract}
We study the well-posedness and dynamic transitions of the surface tension driven convection in a three-dimensional (3D) rectangular box with  non-deformable upper surface and with free-slip boundary conditions. It is shown that as the Marangoni number crosses the critical threshold, the system always  undergoes a dynamic transition. In particular, two different scenarios are studied.  In the first scenario, a single mode losing its stability at the critical parameter gives rise to either a Type-I (continuous) or  a Type-II (jump)  transition. The type of transitions  is dictated by the sign of a computable  non-dimensional parameter, and the  numerical computation of this parameter  suggests that a Type-I transition is favorable.
The second scenario  deals with the case where the geometry of the domain allows two critical modes which possibly characterize a  hexagonal pattern. In this case we show  that the transition can only be either a Type-II or a Type-III (mixed) transition depending on another computable non-dimensional parameter. We only encountered Type-III transition in our numerical calculations. The second part of the paper deals with the  well-posedness and existence of global attractors for the problem.  
\end{abstract}

\maketitle

\section{Introduction}
Since B\'{e}nard's original experiments in 1900 \cite{Benard1900}, it is
well known that when a motionless liquid layer is heated from below, the
liquid layer undergoes a transition from the motionless state to a
convective state as the vertical temperature gradient exceeds a critical
value. If the liquid layer has an upper surface open to ambient air, both
buoyancy and surface tension forces will result from the temperature
gradient. To ensure the onset of convection, these forces must exceed the
dissipative effects of viscous and thermal dissipation \cite{rayleigh}. 
Hence there
exists a critical value of a dimensionless parameter, the Marangoni number, for 
convection to occur.
The effect of the surface tension is dominating in sufficiently shallow layers 
and a micro-gravity environment while buoyancy is
the driving mechanism in deep layers as well as when there is no free
surface.

There are numerous numerical, analytical and experimental studies on the 
B\'{e}nard-Marangoni problem. For a detailed review of the problem, see  
Colinet et al \cite{colinet2001}, Koschmieder \cite{kosch}, Dauby et al. 
\cite{Dauby1993}, Dijkstra \cite{Dijkstra1998} and Rosenblat et al. 
\cite{Rosenblat1982}. 
The main objective of this article is to rigorously investigate  the 
pure surface tension driven convection   in a three-dimensional (3D) rectangular 
box with  non-deformable upper surface  and with free-slip boundary conditions.  

The study is based on the dynamic transition theory developed recently  by Ma  and Wang \cite{ptd, Ma2010}. 
The main philosophy of this theory is to search for  the full set of  transition states, giving a complete characterization on stability and  transition. The set of transition states is often represented by a local attractor. Following this philosophy, the dynamic transition theory is developed  to identify the transition states and to classify them both dynamically and physically.  
One important ingredient  of the theory is the introduction of a new classification scheme of transitions, with which phase transitions are classified into three types: Type-I, Type-II and Type-III. In more mathematically intuitive terms,  they are   called continuous, jump and mixed transitions respectively.  Basically, as the control parameter passes the critical threshold,  the transition states stay in a close neighborhood of the basic state for a Type-I transition,  are outside of a neighborhood of the basic state for a Type-II (jump) transition. For the Type-III transition, a neighborhood is divided into two open regions with a Type-I transition in one region, and a Type-II transition in the other region. 

For the pure surface tension driven convection problem, first we show that  as the Marangoni number crosses the critical threshold, the system always  undergoes a dynamic transition. To classify the type of transitions, and structure of the transition solutions, we consider two scenarios.  In the first scenario, a single mode losing its stability at the critical value of the parameter gives rise to either a Type-I (continuous) or  a Type-II (jump)  transition. The type of transitions  is dictated by the sign of a computable  non-dimensional parameter, and the  numerical computation of this parameter  suggests that a Type-I transition is favorable for all values of the Prandtl number.
The second scenario  deals with the case where the geometry of the domain allows two critical modes which possibly characterize a  hexagonal pattern. In this case we show  that the transition can only be either a Type-II or a Type-III (mixed) transition depending on another computable non-dimensional parameter. However we only encountered Type-III transition in our numerical calculations. 

One crucial part  of the analysis is the reduction of the original partial differential equation system to the center manifold generated by the first unstable modes, leading to either a one-dimensional (for the first scenario) or a two-dimensional (for the second scenario) dynamical system. The dynamic transition behavior is then characterized using the reduced system following the ideas from the dynamic transition theory. However, it is worth mentioning that for the hexagon case, the reduced two-dimensional system consists of both quadratic and cubic nonlinearities. In addition, the quadratic terms are degenerate, and the cubic terms are needed to fully characterize the flow structure. This type of nonlinearities with degenerate quadratic terms appear also in many other fluid mechanical problems. Here for the first time, we are able to fully characterize the dynamic transitions. In particular, in the Type-III transition case, the hexagonal flows are represented by the local attractor for the continuous transition part of the Type-III transition. Furthermore, these hexagonal  flow patterns  are metastable. Namely, the original system undergoes a dynamic transition either to these hexagonal structure or to some more complicated flow pattern far away from the basic state. 

The paper is organized as follows. The model is given in Section 2. Section 3 deals with the linear stability problem, leading to precise information on the principle of exchange of stabilities of the problem. Section 4 reduces the original problem to the center manifold generated by the first unstable modes, and Section 5 states the main dynamic transition theorems, which are proved in Section 6. In section 7, the well-posedness of the problem  and existence of global attractors are studied. As we know,  well-posedness is one of the basic issues for nonlinear problems, and the existence of global attractor indicates that the system is a dissipative system in the sense of Prigogine.  A short summary and a discussion of the results in relation to  the physics of the problem
is presented in Section 8.

\section{The Model}
With the Boussinesq approximation, the (non-dimensional) equations governing the motion and states  of the pure
Marangoni convection on a nondimensional  rectangular domain
$\Omega =\left( 0,L_{1}\right)
\times \left( 0,L_{2}\right) \times \left( 0,1\right) \subset \mathbb{R}^{3}$
are given as follows 
(see e.g., Dijkstra \cite{Dijkstra1998}): 
\begin{equation}  \label{main}
\begin{aligned} & \frac{\partial \mathbf{u}}{\partial t}+\left(
\mathbf{u}\cdot \nabla \right) \mathbf{u} = \text{Pr}\left( -\nabla p+\Delta
\mathbf{u}\right) , \\ & \frac{\partial \theta }{\partial t}+\left(
\mathbf{u}\cdot \nabla \right) \theta = w+\Delta \theta , \\ & \nabla \cdot
\mathbf{u} =0, \\ & \mathbf{u}\left( 0\right) =\mathbf{u}_{0}\text{, \ \
}\theta \left( 0\right) =\theta _{0}. \end{aligned}
\end{equation}
Here the effect of gravity is ignored by setting the Rayleigh number $Ra=0$, and we consider the deviation 
from a motionless basic steady state with  a constant vertical temperature
gradient $\eta >0$. The unknown functions are the  velocity field $\mathbf{u}=\left( u,v,w\right) $,  the temperature function $\theta $,  and the pressure function $p$. In addition, $\text{Pr}={\nu }/{\kappa }>0$ is the Prandtl
number, $\nu $ stands for the kinematic viscosity,  and $\kappa $ is  the thermal
diffusivity. 

The above system is supplemented with a set of boundary conditions. We use the free-slip boundary conditions on the 
lateral boundaries, and the rigid (no slip) boundary condition and perfectly conducting on the bottom boundary.  The top
surface is assumed to be a non-deformable free surface with  a surface
tension of the form 
$$
\xi=\xi_{0} (1 - \gamma_{T}\theta).
$$
Namely, the boundary conditions are as follows:
\begin{equation}  \label{bc}
\begin{aligned} 
& u =\frac{\partial v}{\partial x}=\frac{\partial w}{\partial x}=\frac{\partial \theta }{\partial
x}=0\,
      &&\text{at}\,x=0,L_{1}, \\ & \frac{\partial u}{\partial y}
=v=\frac{\partial w}{\partial y}=\frac{\partial \theta }{\partial
y}=0\,&&\text{at}\,y=0,L_{2}, \\ & u =v=w=\theta =0\,&&\text{at}\,z=0, \\ &
\frac{\partial \left( u,v\right) }{\partial z}+\lambda \nabla _{H}\theta
=w=\frac{\partial \theta }{\partial z}+ \text{Bi} \theta
=0\,&&\text{at}\,z=1, 
\end{aligned}
\end{equation}
where $\nabla _{H}=\left( \partial _{x},\partial _{y}\right) $,  $\text{Bi} \geq 0$ 
is the Biot number, and the Marangoni number $\lambda $ is the
control parameter defined by 
\begin{equation*}
\lambda =\frac{\xi_{0}\gamma _{T}\eta d^{2}}{\rho _{0}\nu \kappa }>0,
\end{equation*}
$d$ is the dimensional depth of the box and $\rho _{0}$ is the
reference value for the density. Note that the Marangoni number represents the ratio of
the destabilizing surface tension gradient to the stabilizing forces
associated with thermal and viscous diffusion.

\section{Principle of Exchange of Stability}
We  recall in this section the linear theory of the problem. 
The linear equations associated with (\ref{main})  and (\ref{bc}) are:
\begin{equation}
\begin{aligned} 
& \text{Pr}\left( -\nabla p+\Delta \mathbf{u}\right) =\beta \mathbf{u}, \\ 
& w+\Delta \theta =\beta \theta , \\ & \nabla \cdot
\mathbf{u} =0,  
\end{aligned}  \label{Lin-0}
\end{equation}
supplement with the same boundary conditions given by (\ref{bc}).
The adjoint problem defined by
$$ 
\int_\Omega (L_{\lambda}\psi) \overline{\psi^{\ast}}=\int_\Omega \psi (\overline{L_{\lambda}^{\ast} \psi^{\ast}} )
$$
can be written as:
\begin{equation}  \label{adj1}
\begin{aligned} & \text{Pr}\left( -\nabla p^{\ast }+\Delta \mathbf{u}^{\ast
}\right) +\overrightarrow{k}\theta ^{\ast } =\overline{\beta
}\mathbf{u}^{\ast }, \\ & \Delta \theta ^{\ast } =\overline{\beta }\theta
^{\ast }, \\ & \nabla \cdot \mathbf{u}^{\ast } =0. \end{aligned}
\end{equation}
Here an overbar denotes complex conjugation. 
The boundary conditions for the adjoint problem at the lateral sides and at $z=0$ 
are the same as the linear eigenvalue problem but are different at $z=1$:
\begin{equation}
\frac{\partial u^{\ast }}{\partial z}=\frac{\partial v^{\ast }}{\partial z}
=w^{\ast }=0, \quad \frac{\partial \theta ^{\ast }}{\partial z}+\text{Bi}\theta
^{\ast }+\lambda \text{Pr}\frac{\partial w^{\ast }}{\partial z}=0 \quad \text{at}
\quad z=1.  \label{adjbc}
\end{equation}

By the separation of variables, we represent the solutions 
in the following form:
\begin{equation} \label{sepofvar}
\begin{aligned} 
& u_I =U_I\left( z\right) \sin L _{1}^{-1}I_x\pi x\cos L _{2}^{-1}I_y\pi y, \\ 
& v_I =V_I\left( z\right) \cos L _{1}^{-1}I_x\pi x\sin L _{2}^{-1}I_y\pi y, \\ 
& w_I =W_I\left( z\right) \cos L _{1}^{-1}I_x\pi x\cos L _{2}^{-1}I_y\pi y, \\ 
& \theta_I =\Theta_I \left( z\right) \cos L _{1}^{-1}I_x\pi x\cos L _{2}^{-1}I_y\pi y. 
\end{aligned}  
\end{equation} 
for $I=(I_x,I_y)\in \mathbb{Z}\times\mathbb{Z}$.
Let
\begin{equation}\label{alpha}
\alpha_I=\left((L _{1}^{-1}I_x)^2+(L _{2}^{-1}I_y )\right)^{1/2}\pi.
\end{equation}
It is easy to see that, for $\alpha_I \neq 0$, the horizontal velocity components can be obtained as:
\begin{equation}\label{hor.comp}
U_I (z)=-\frac{L _{1}^{-1}I_x\pi}{\alpha_I^2} DW_I(z), \qquad  V_I(z)=-\frac{L _{2}^{-1}I_y \pi}{\alpha_I^2} DW_I(z),
\end{equation}
Here
$$
D=\frac{d}{dz}.
$$
By \eqref{sepofvar} and \eqref{hor.comp} for $I=(I_x,I_y)$, $\phi_I$  does not change when $I_x$ or $I_y$ changes sign. Thus we consider only nonnegative wave indices $I_x,I_y\in \{ 0,1,2,\dots \}$. When $I_x^{2}+I_y^{2}\neq 0$, the ODE satisfied by $W$ and $\Theta$ is
\begin{equation} \label{Lin3}
\begin{aligned} 
& \left( D^{2}-\alpha ^{2}\right) \left( D^{2}-\alpha^{2}-\text{Pr}^{-1}\beta \right) W =0, \\ 
& \left( D^{2}-\alpha ^{2}-\beta \right) \Theta =-W. 
\end{aligned} 
\end{equation}
Using the divergence free condition, we can write the boundary conditions (\ref{bc}) 
at the upper and lower boundaries as: 
\begin{equation}
\begin{aligned} 
& W\left( 0\right) =DW\left( 0\right) =\Theta \left(0\right) =0, \\ 
& W\left( 1\right) =D\Theta \left( 1\right) +\text{Bi}\Theta
\left( 1\right) =D^{2}W\left( 1\right) +\alpha ^{2}\lambda \Theta \left(
1\right) =0. 
\end{aligned}  \label{linbc}
\end{equation}
We write the corresponding ordinary differential equations to the adjoint equations as
\begin{equation} \label{adj-equ}
\begin{aligned} 
& \left( D^{2}-\alpha ^{2}\right) \left( D^{2}-\alpha^{2}-\text{Pr}^{-1}\bar{\beta} \right) W^{\ast }
-\alpha ^{2}\text{Pr}^{-1}\Theta^{\ast } =0, \\ 
& \left( D^{2}-\alpha ^{2}-\bar{\beta} \right) \Theta ^{\ast} =0,
\end{aligned}
\end{equation}
with the boundary conditions 
\begin{equation} \label{adjodebc}
W^{\ast }\left( 1\right) =D^{2}W^{\ast }\left( 1\right) =D\Theta ^{\ast}\left( 1\right) 
+\text{Bi}\Theta ^{\ast }\left( 1\right) +\lambda
\text{Pr } DW^{\ast }\left( 1\right) =0.
\end{equation}

Solving the boundary condition 
$D^{2}W\left( 1\right) +\alpha ^{2}\lambda \Theta \left( 1\right) =0$ 
in the $\beta=0$ case gives us the critical Marangoni number:
\begin{equation}\label{Mac}
\lambda _{c}=\min_{\substack{ j,k\in \left\{ 0,1,2,\dots \right\}  \\
j^{2}+k^{2}\neq 0  \\ \alpha ^{2}=j^{2}\pi ^{2}L _{1}^{-2}+k^{2}\pi ^{2}L
_{2}^{-2}}}\frac{8\alpha \left( \alpha \cosh \alpha +\text{Bi}\sinh \alpha
\right) \left( \alpha -\cosh \alpha \sinh \alpha \right) }{\alpha ^{3}\cosh
\alpha -\sinh ^{3}\alpha }.  
\end{equation}%
Denote the set of critical indices by $\mathcal{C}$:
\begin{equation}\label{P}
\mathcal{C}=\{I\mid\text{$I=(I_x,I_y)$ minimizes \eqref{Mac}}\}.
\end{equation}
Since the function being minimized at \eqref{Mac} is a convex function of $\alpha$, as shown in Figure \ref{Biot-alpha}, the set $\mathcal{C}$ is non-empty. Clearly $\mathcal{C}$ is finite.
\begin{figure}
\includegraphics[scale=0.6]{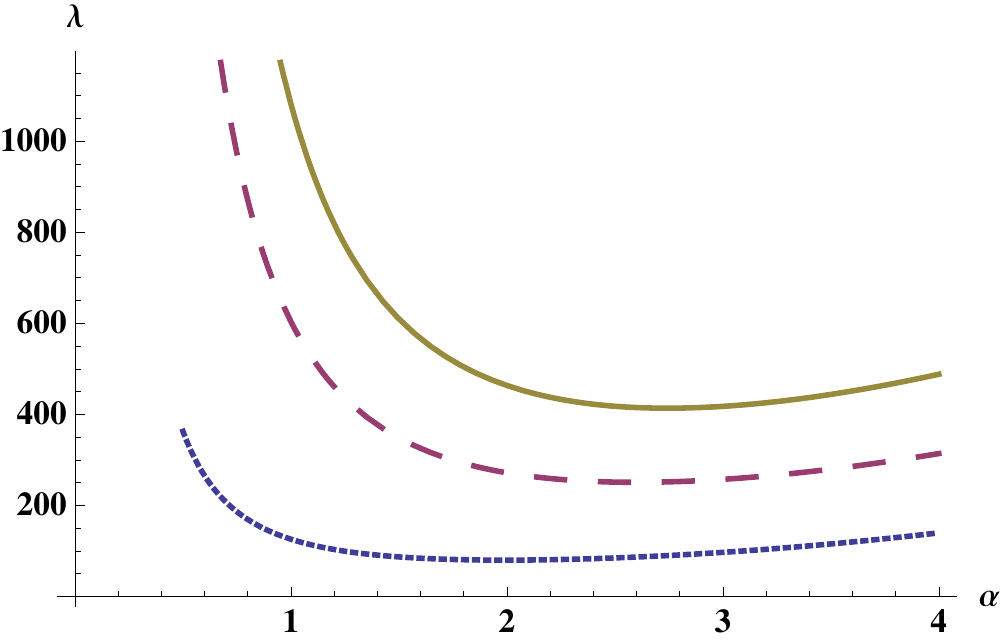}
\caption{Marginal stability curves at $Bi=0$ (dotted), $Bi=5$ (dashed), $Bi=10$(continuous). }
\label{Biot-alpha}
\end{figure}

For a fixed $I=(I_x,I_y)$, there are infinitely many eigenvalues which can be ordered as:
$$
\Re{\beta_{(I,1)}}\geq\Re{\beta_{(I,2)}}\geq\cdots,
$$
and we denote the corresponding eigenvectors by $\phi_{(I,k)}$. Then the critical modes are $\phi_{(I,1)}$ where $I\in\mathcal{C}$.

Using a Green-functions technique to reduce \eqref{Lin3}-\eqref{linbc} to a single differential equation,  
Vrentas-Vrentas \cite{Vrentas2004} gave a simple analytical proof that for the problem \eqref{Lin3}-\eqref{linbc}, critical eigenvalues must be real. Thus $\beta_{(J,1)}\in\mathbb{R}$ for $J\in\mathcal{C}$. Moreover the following theorem justifies   the principle of  exchange of stability.

\begin{theorem}\label{PES} 
The eigenvalues of the linear
problem (\ref{Lin-0}) with boundary conditions (\ref{bc}) satisfy%
\begin{align}
& \beta _{(J,1)} \left( \lambda \right) =\left\{ 
\begin{aligned}
& <0 && \lambda <\lambda _{c} \\ 
& =0 && \lambda =\lambda _{c} \\ 
& >0 && \lambda >\lambda _{c},
\end{aligned}\right. 
          &&  \forall J\in\mathcal{C}\label{PES1} \\
& \Re \beta _{(J,k)}\left( \lambda _{c}\right) <0 && \forall J\notin \mathcal{C}.
\label{PES2}
\end{align}
where $\mathcal{C}$ is given by \eqref{P}.
\end{theorem}

\begin{proof}
Let $J \in \mathcal{C}$. We will use the following formula, derived by Ma and Wang \cite{ptd}, to verify the
critical crossing of the first eigenvalue:
\begin{equation*}
\frac{d}{d\lambda }\beta _{(J,1)}\left( \lambda _{c}\right) =\left(\frac{d}{%
d\lambda } L_{\lambda _{c}}\phi _{(J,1)},\phi _{(J,1)}^{\ast }\right) .
\end{equation*}
Here $L_{\lambda }$ is given by (\ref{linearop}). Thus%
\begin{eqnarray*}
\frac{d}{d\lambda }\beta _{(J,1)}\left( \lambda _{c}\right) &=&\text{Pr}%
\int_{z=1}\nabla _{H}\theta _{(J,1)}\cdot \left( u_{(J,1)}^{\ast },v_{(J,1)}^{\ast
}\right) dxdy \\
&=&\frac{L_{1}L_{2}}{4}\Theta _{(J,1)}\left( 1\right) DW_{(J,1)}^{\ast }\left(
1\right) \\
&=&\frac{L_{1}L_{2}}{4}\frac{\left( -\alpha_J ^{3}\cosh \alpha_J +\sinh
^{3}\alpha_J \right) ^{2}}{\sinh \alpha_J\left( \alpha_J \cosh \alpha_J +\text{Bi}\sinh \alpha_J \right) } \\
&>&0.
\end{eqnarray*}
When integrating the boundary integral above, without loss of generality, we assumed $I_x\neq0$ and $I_y\neq0$. So (\ref{PES1}) is valid, and  (\ref{PES2}) is a simple consequence of (\ref{Mac}).
\end{proof}

\section{Transition Equations}
The dynamic transition theory developed by Ma and Wang \cite{ptd} implies that as soon as the linear problem indicates an instability, the nonlinear system always undergoes a dynamic transition, leading to one of the three type of transitions, 
Type-I, II and III. The type of transitions is dictated by the nonlinear interactions. For this purpose,  we will follow the method developed by Ma  and Wang \cite{ptd}, which relies
heavily on the reduction of the problem to the center manifold in the first
unstable eigenmodes. The key step is to find the approximation of the
reduction to certain order, leading to a \textquotedblleft
nondegenerate\textquotedblright\ system with higher order perturbations. The
full dynamic transition and stability analysis is  then carried out.

Consider the critical set $\mathcal{C}$ given by \eqref{P} and let $$\phi=\sum_{I\in \mathcal{C}}y_I \phi_{(I,1)}+\Phi(y)$$ where $\Phi$ is the center manifold function, $\phi_{(I,1)}$ are the critical (first) eigenvectors and $y_I \in \mathbb{R}$ are the corresponding amplitudes. Multiplying the governing evolution equation by $\phi_{(I,1)}^{\ast}$, we see that the amplitude of the critical modes satisfies the following transition equation
\begin{equation}\label{reduced1}
\frac{dy_I}{dt}=\beta_{(I,1)}( \lambda ) y_I+\frac{1}{\left\langle \phi
_{(I,1)},\phi _{(I,1)}^{\ast }\right\rangle }\left\langle G\left( \phi ,\phi \right)
,\phi _{(I,1)}^{\ast }\right\rangle .
\end{equation}
Here $\beta _{(I,1)}\left( \lambda \right) $ is the eigenvalue corresponding to $%
\phi _{I,1}$ which satisfies \eqref{PES1}. The pairing $\left\langle \cdot ,\cdot \right\rangle $
denotes the $L^{2}$-inner product over $\Omega $. The bilinear operator $G$
is as defined by (\ref{bilinearop}).
We write the phase space as 
\begin{equation*}
H=E_{1}\oplus E_{2}, \qquad E_{1}=\text{span}\{\phi_{(I,1)}\mid I\in \mathcal{C}\}, \qquad E_{2}=E_{1}^{\perp }.
\end{equation*}
We use the following approximation of the center manifold function derived by  Ma and Wang \cite{ptd}:
\begin{equation}
-\mathcal{L}_{\lambda }\Phi (\sum_{I\in \mathcal{C}}y_I \phi_{(I,1)}, \lambda) =P_{2}G(\sum_{I\in \mathcal{C}}y_I \phi_{(I,1)}) +o(2),  
\label{centerformula}
\end{equation}
where $\mathcal{L}=L_{\lambda }\mid _{E_{2}}$ is the restriction of the
linear operator defined by (\ref{linearop}) onto $E_{2}$, $P_{2}$ is the
projector from $H$ onto $E_{2}$ and by $o\left( 2\right) $ we mean 
\begin{equation*}
o(2)=o( \left\vert y\right\vert ^{2}) +O( \left\vert
y\right\vert ^{2}\left\vert \beta _{(I,1)}\left( \lambda \right) \right\vert), \qquad I\in\mathcal{C}.
\end{equation*}%
By \eqref{sepofvar}, there is a finite set of indices $\mathcal{S}$ which depends on $\mathcal{C}$, such that 
\begin{equation}\label{ortho}
\begin{aligned}
& \left\langle G\left( \phi _{(I,1)},\phi _{(J,1)}\right) ,\phi_{(K,k)}^{\ast}\right\rangle =0,
\qquad \forall K\notin\mathcal{S}, \, \forall I,J\in\mathcal{C}, \forall k\geq1.
\end{aligned}
\end{equation}
The set $\mathcal{S}$ can be precisely defined as:
\begin{equation}\label{S}
\mathcal{S}=\{(K_x,K_y)\mid K_i=|I_i\pm J_i|,\, i=x,y, \, (I_x,I_y)\in\mathcal{C},(J_x,J_y)\in\mathcal{C}\}.
\end{equation}
Multiplying \eqref{centerformula} by $\phi_{M,k}^{\ast}$ and using \eqref{ortho}, 
we obtain the following approximation of the center manifold function:
\begin{equation}\label{center}
\Phi =\sum_{\substack{I,J\in \mathcal{C},\\ K\in\mathcal{S}\cap\mathcal{C},k\geq2\\ \text{or }K\in\mathcal{S}\setminus\mathcal{C},k\geq1}}y_I y_J \Phi_{IJK}^k\phi_{(K,k)}+o(2).
\end{equation}
The computation of the coefficients $\Phi_{IJK}^k$ in \eqref{center} is tedious yet straightforward. Using \eqref{ortho} and \eqref{center}, the equation (\ref{reduced1}) becomes for each $I \in \mathcal{C}$,
\begin{align}
\frac{dy_I}{dt}=
&\beta_{I,1}(\lambda) y_I+\sum_{\substack{J,K \in  \mathcal{C}}} G((J,1),(K,1),(I,1))\, y_J y_K 
      \label{reduced2}\\
&+\sum_{\substack{J,L,M\in \mathcal{C},\\ K\in\mathcal{S}\cap\mathcal{C},k\geq2\\ \text{or }K\in\mathcal{S}\setminus\mathcal{C},k\geq1}}H((J,1),(K,k),(I,1)) \Phi_{LMK}^k y_J y_L y_M+o(3), \nonumber 
\end{align} 
where 
\begin{equation}
\begin{aligned}
& G((J,j),(K,k),(I,i))=\frac{1}{\langle \phi _{(I,i)},\phi_{(I,i)}^{\ast} \rangle}\langle G(\phi_{(J,j)},\phi_{(K,k)}),\phi_{(I,i)}^{\ast}\rangle, \\
& H((J,j),(K,k),(I,i))= G((J,j),(K,k),(I,i))+G((K,k),(J,j),(I,i)).
\end{aligned}
\end{equation}

The type of transitions  is then determined by the equation \eqref{reduced2}. 
The coefficients of \eqref{reduced2} depend on the  parameters Bi, Pr 
and the structure of the critical index set $\mathcal{C}$.
These coefficients  can be derived using  the formula \eqref{centerformula}. In particular, 
if $\beta_{(M,m)}$ is a real simple  eigenvalue,  then the $\phi_{(M,m)}$ component of the center manifold is simply:
\begin{equation} \label{Phi_real}
\Phi_{IJM}^{m}=-\frac{\langle G\left( \phi _{(I,1)},\phi _{(J,1)}\right)
,\phi_{(M,m)}^{\ast }\rangle }{\beta _{(M,m)} \langle \phi_{(M,m)},\phi_{(M,m)}^*\rangle}.
\end{equation}

\section{Dynamic Transitions}
We study dynamic transitions of the system in two different scenarios. The first is the case where the critical eigenvalue is simple, and  there is only one critical index $I=(I_x,I_y)\in \mathcal{C}$. In this case,  the quadratic terms vanish $\delta_{JKI}=0$. Moreover the set of indices for quadratic interactions are $\mathcal{S}=\{(0,0),(2I_x,0),(0,2I_y),(2I_x,2I_y)\}$. Thus transition equation \eqref{reduced2} becomes
\begin{equation}\label{1inter}
\frac{dy_I}{dt}=\beta_{(I,1)}(\lambda) y_I+ c_I y_I^3+o(3) .
\end{equation}
where 
\be \label{sigma}
c_I=\sum_{K\in\mathcal{S}, k\geq1}H((I,1),(K,k),(I,1))\Phi_{IIK}^k.
\ee

The dynamical transition is then characterized by the following  theorem:

\begin{theorem}\label{single.thm}
Let $c_I$ be as in \eqref{sigma}.
Assume there is a single critical index $I\in\mathcal{C}$. If $c_I<0$ then the system undergoes a Type-I transition at $\lambda=\lambda_c$. In particular, the system bifurcates to  two steady state solutions for $\lambda>\lambda_c$ which are local attractors and are given by:
\begin{equation}\label{1interstructure}
\psi(\lambda)=\pm\sqrt{-\frac{\beta_{(I,1)}(\lambda)}{c_I}}\phi_{(I,1)}+o(\beta_{(I,1)}^{1/2}).
\end{equation}
If $c_I>0$, then the transition is Type-II and the system bifurcates to two steady state solutions for $\lambda<\lambda_c$ which are repellers and there are no steady state solutions that bifurcate from $(\psi,\lambda)=(0,\lambda_c)$ for $\lambda>\lambda_c$.
\end{theorem}

The second scenario deals with hexagonal  patterns which may be expected when two modes characterizing a hexagon becomes unstable at the same critical parameter. For this to happen, a certain relation has to hold for the length scales of the box:
\begin{equation} 
\label{hexbox}
\frac{L_1}{L_2}=\frac{I_x}{I_y \sqrt{3}}, 
\end{equation}
where $I_x$ and $I_y$ are two positive integers. With this relation in hand, the modes with indices $I=(I_x,I_y)$ and $J=(0,2I_y)$ have the same wave number, $\alpha_I=\alpha_J$. Moreover, the vector field $Y \phi_I+Z \phi_J$ defines a hexagon parallel to $y$ when $Y=\pm 2 Z$; see Figure \ref{hexa}. Thus a pair of modes will be critical at the critical parameter $\lambda_c$ whenever one of the modes minimizes the relation \eqref{Mac}. In this case, the reduced transition equations are a two-dimensional system on the center manifold generated by the two unstable modes. To state the transition theorem, we define two crucial parameters and leave the reduced equation and other parameters in in the proof of the transition theorem in the next section:
\begin{equation}
\label{2p}
\begin{aligned} 
& a_1=H((J,1),(I,1),(I,1)), \\
& b_2=\sum_{l\geq1 \text{ and } S=0,2J} H((J,1),(S,l),(J,1))  \Phi_{JJS}^l.
\end{aligned}
\end{equation}

Note that $b_2=c_J$ where $c_I$ is as in \eqref{sigma} (with $J$ should be written instead of $I$ in \eqref{sigma}).

\begin{figure}[tbp]
\includegraphics[scale=0.6]{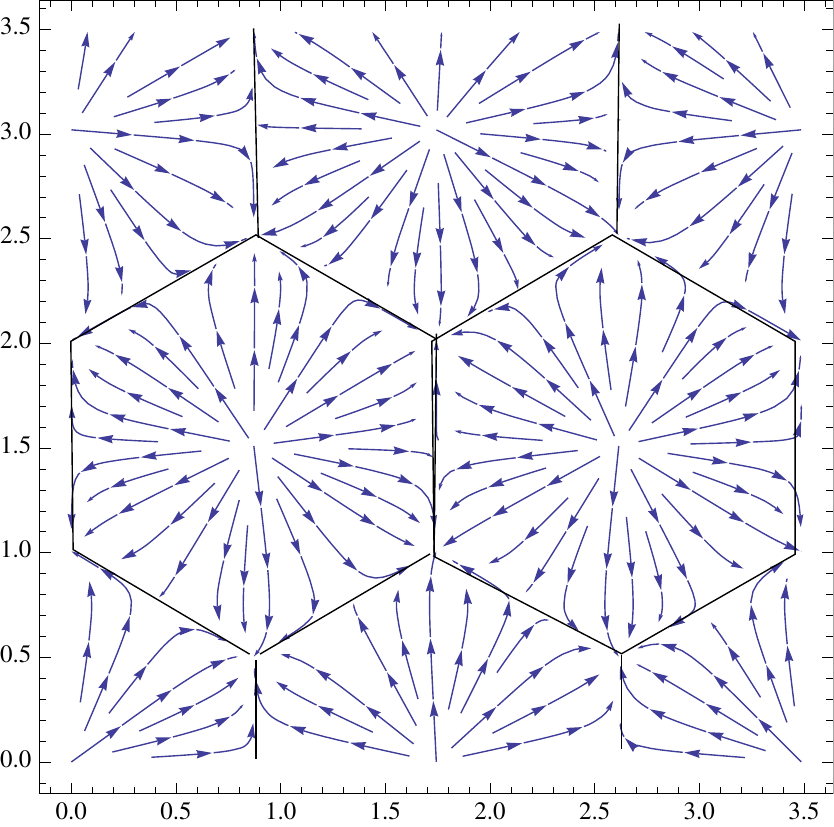}
\caption{The flow structure of $2\phi_{2,1}+\phi_{0,2}$ at $z=1$ for the box dimensions $L_1=L_2=3.5$}
\label{hexa}
\end{figure}

\begin{figure}[ht]
\centering
\subfigure[$\lambda<\lambda_c$]{
\includegraphics[scale=0.4]{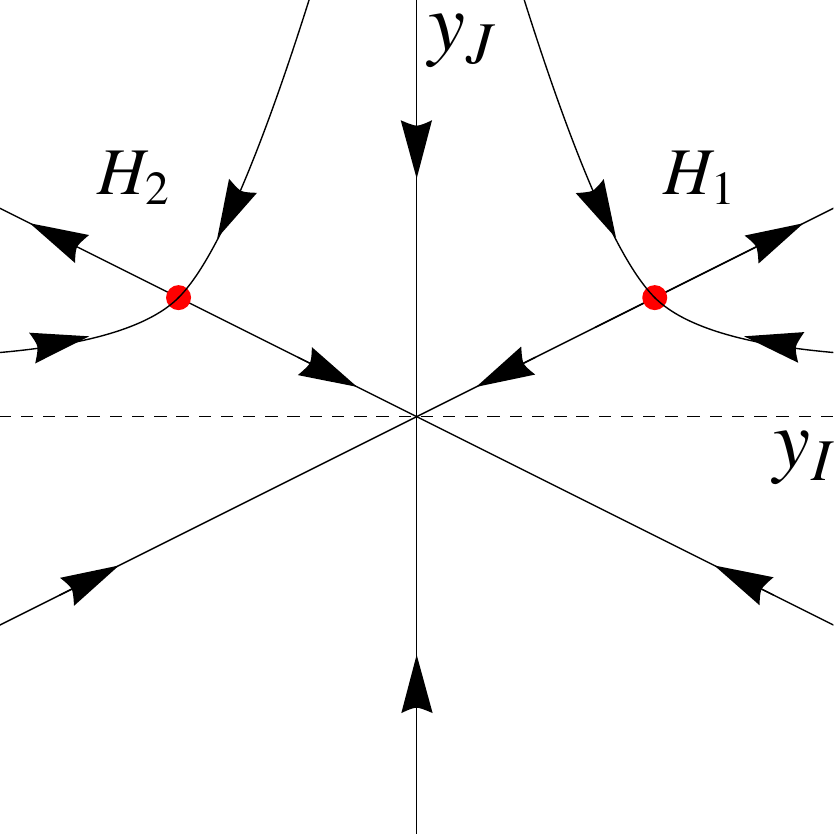}
\label{fig:subfig1}
}
\subfigure[$\lambda=\lambda_c$]{
\includegraphics[scale=0.4]{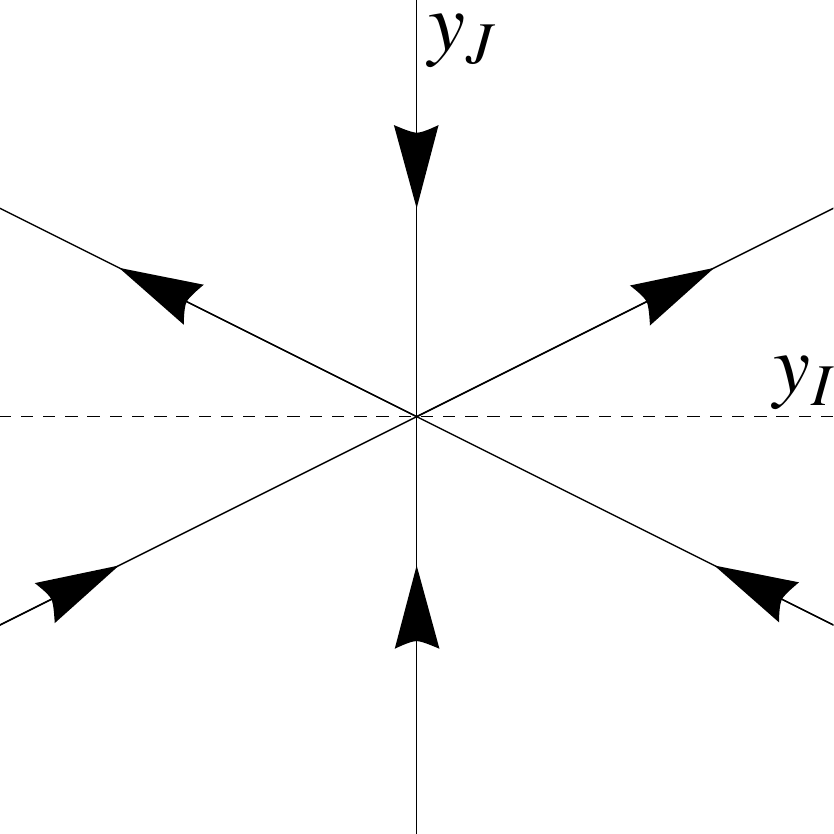}
\label{b2<0cr}
}
\subfigure[$\lambda>\lambda_c$]{
\includegraphics[scale=0.4]{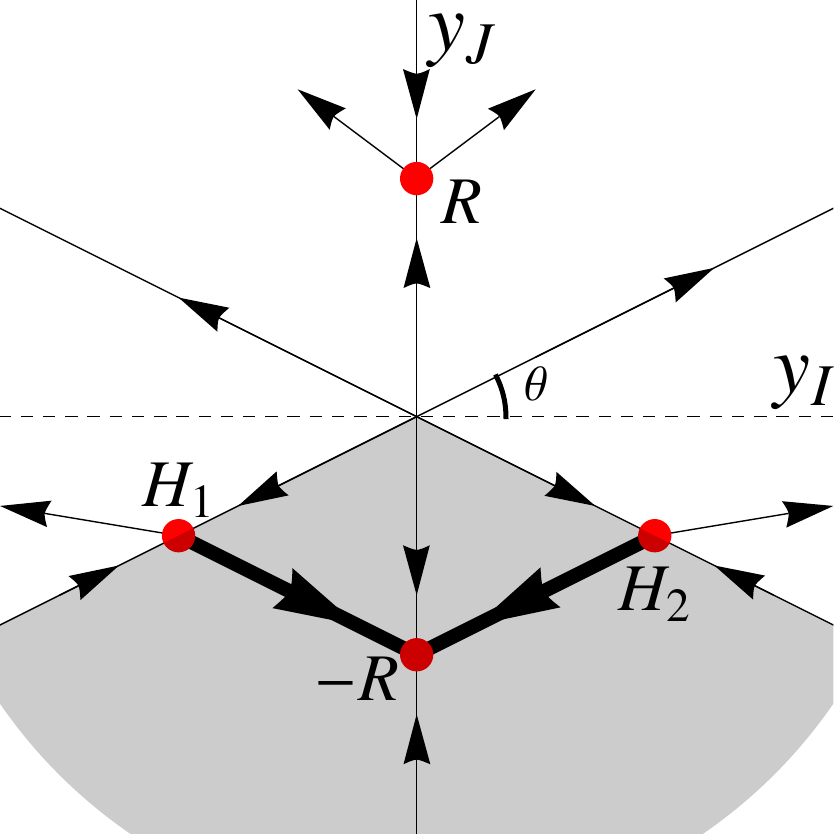}
\label{attract}
}
\caption[Optional caption for list of figures]{The transition for $b_2<0, a_1>0$ is Type-III. Here $H_1$  and  $H_2$ represent the steady states having hexagonal patterns (Figure \ref{hexa}) and $\pm R$ represent the steady states with roll patterns (Figure \ref{1interactionstructure}). The bifurcating attractor $\Sigma_{\lambda}$ on $\lambda>\lambda_c$, given in  Theorem \ref{hex.thm}  and  shown in bold in (c),  contains these steady states and the connecting  heteroclinic orbits. This attractor has the shaded sector $(\pi+\theta,2\pi-\theta)$ as the basin of attraction where $\theta=\arctan1/2$. }
\label{b2less0}
\end{figure}

\begin{figure}[ht]
\centering
\subfigure[$\lambda<\lambda_c$]{
\includegraphics[scale=0.4]{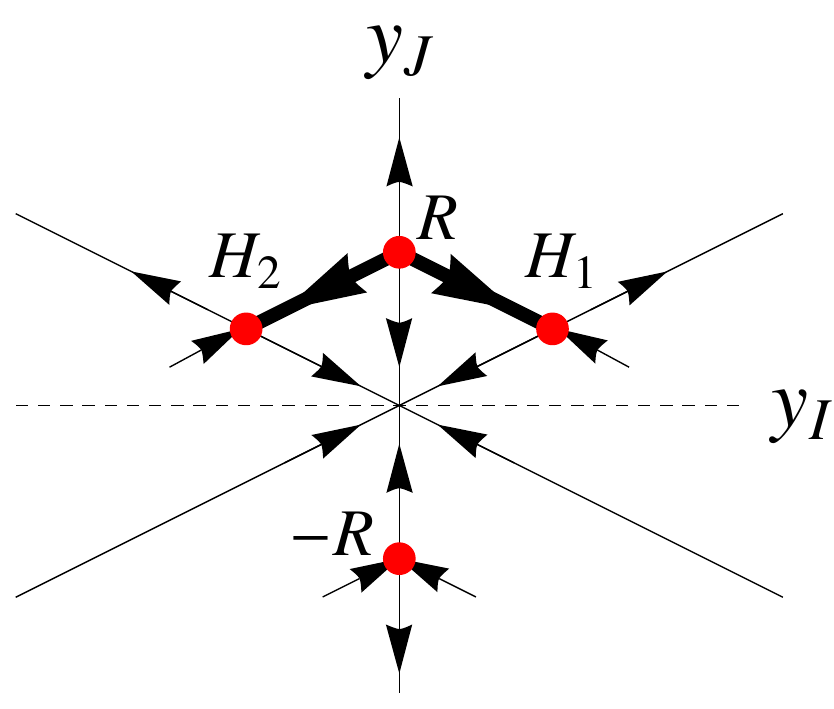}
\label{repel}
}
\subfigure[$\lambda=\lambda_c$]{
\includegraphics[scale=0.4]{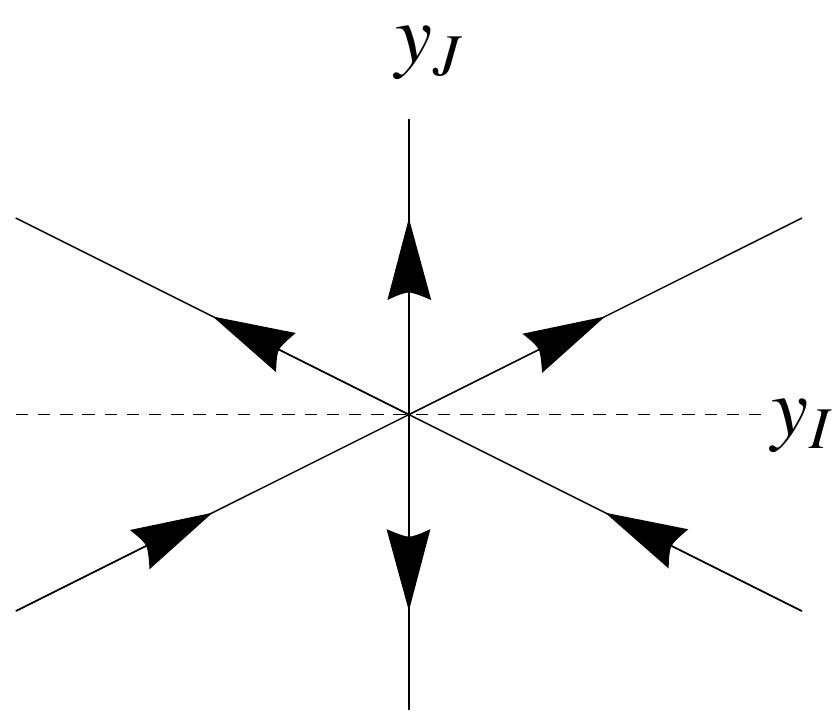}
}
\subfigure[$\lambda>\lambda_c$]{
\includegraphics[scale=0.4]{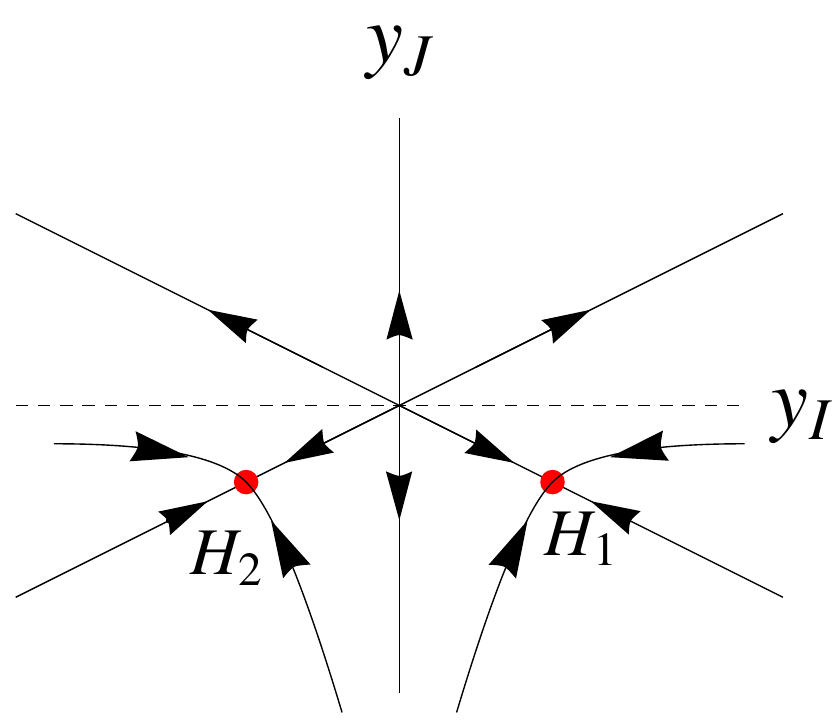}
}
\caption{The transition for $b_2>0, a_1>0$ is Type-II. The bifurcating repeller $\Sigma_{\lambda}$ on $\lambda<\lambda_c$ is shown in bold. In this case, the system undergoes a drastic transition leading to more complex patterns in another attractor away from the basic motionless state. The structure of this  attractor needs to be determined by different methods.}
\label{b2>0}
\end{figure}

\begin{theorem} \label{hex.thm}
Assume that the horizontal length scales satisfy the relation \eqref{hexbox} for some positive integers $I_x,I_y$ and the critical index set is 
$$
\mathcal{C}=\{I=(I_x,I_y),J=(0,2I_y)\}.
$$ 
Let the two parameters $b_2$ and $a_1$ be defined by \eqref{2p}. Assume $a_1>0$\footnote{In the case where $a_1<0$ the assertions given by Theorem \ref{hex.thm} hold true with the regions and the steady states flipped with respect to the $y_I$ axis.}, and 
\begin{equation}\label{steady}
\begin{aligned}
\pm &  R^{\lambda}=(y_I,y_J)=\pm(0, \sqrt{\frac{\beta(\lambda)}{-b_2}})+O(\beta(\lambda)), \\
& H_i^{\lambda}=\frac{\beta(\lambda)}{a_1}(2(-1)^i ,-1)+O(\beta(\lambda)^2), \,i=1,2, \\
\end{aligned}
\end{equation}
\begin{itemize}
\item[i)]
If $b_2<0$ then the system undergoes a Type-III transition at $\lambda=\lambda_c$ and the following assertions hold true:
\begin{itemize}
\item[a)] The topological structure of the transition is as in Figure \ref{b2less0}.
\item[b)] There is a neighborhood $\mathcal{U}$ of $\phi=0$ in $H$ such that for any $\lambda_c<\lambda<\lambda_c+\epsilon$ with some $\epsilon>0$, $\mathcal{U}$ can be decomposed into two open sets $\mathcal{U}_1^{\lambda}$, $\mathcal{U}_2^{\lambda}$,
$$\overline{\mathcal{U}}=\overline{\mathcal{U}_1^{\lambda}}\cup\overline{\mathcal{U}_2^{\lambda}}, \qquad \mathcal{U}_1^{\lambda}\cap\mathcal{U}_2^{\lambda}=\emptyset$$ such that
\begin{equation*}
\begin{aligned}
& \lim_{\lambda\rightarrow\lambda_c}\limsup_{t\rightarrow\infty} ||S_{\lambda}(t,\varphi)||_H=0 && \qquad \forall \varphi\in\mathcal{U}_1^{\lambda}, \\
& \limsup_{t\rightarrow\infty} ||S_{\lambda}(t,\varphi)||_H\geq\delta>0&& \qquad \forall \varphi\in\mathcal{U}_2^{\lambda},
\end{aligned}
\end{equation*}
for some $\delta>0$. Here $S_{\lambda}$ is the evolution of the solution with initial data $\varphi$.
Moreover $\mathcal{P}(U_1^\lambda)$ is a sectorial region and $\mathcal{P}(U_2^\lambda)$ consists of two sectorial regions given by:
\begin{equation}
\begin{aligned}
& \mathcal{P}(\mathcal{U}_1^\lambda)=\mathcal{U}\cap \{x\in \mathbb{R}^2\mid \pi+\theta<arg(x)<2\pi-\theta\},\\
& \mathcal{P}(\mathcal{U}_2^\lambda)=\mathcal{U}\cap\{x\in \mathbb{R}^2\mid -\theta<arg(x)<\pi/2 \text{ or } \pi/2<arg(x)<\pi+\theta\},
\end{aligned}
\end{equation}
where $\theta=\arctan{1/2}$ and $\mathcal{P}$ is the projection onto $\phi_{(I,1)}$, $\phi_{(J,1)}$ plane.
\item[c)] The system bifurcates to an attractor $\Sigma_{\lambda}$ which consists of three steady states $H_1^{\lambda},H_2^{\lambda},-R^{\lambda}$ and heteroclinic orbits connecting $-R^{\lambda}$ to $H_1^{\lambda}$ and $-R^{\lambda}$ to $H_2^{\lambda}$. Namely, $\Sigma_{\lambda}$ is the arc connecting these three steady states as shown in Figure \ref{attract}, and has basin of attraction $\mathcal{U}_1^{\lambda}$.
\end{itemize}
\item[ii)] If $b_2>0$ then the system undergoes a Type-II transition at $\lambda=\lambda_c$ and the following assertions are true:
\begin{itemize}
\item[a)] The topological structure of the transition is as given by Figure~\ref{b2>0}. \item[b)] There is a bifurcating repeller $\Sigma_{\lambda}$ on $\lambda<\lambda_c$ which consists of three steady states, $H_1^{\lambda}, H_2^{\lambda}, R^{\lambda}$ and the heteroclinic orbits connecting $R^{\lambda}$ to $H_1^{\lambda}$ and $H_2^{\lambda}$ respectively. $\Sigma_{\lambda}$, topologically, is as shown in Figure~\ref{repel}.
\item[c)] Finally for $\lambda_c+\epsilon>\lambda>\lambda_c$ there is an open neighborhood $\mathcal{U}$ of $\phi=0$ and a dense, open subset $\mathcal{U}^{\lambda}$ of $\mathcal{U}$ such that
$$ \limsup_{t\rightarrow\infty} ||S_{\lambda}(t,\varphi)||_H\geq\delta>0, \qquad \forall \varphi\in\mathcal{U}^{\lambda},
$$
for some $\delta>0$.
\end{itemize}
\end{itemize}
\end{theorem}

A few remarks are now in order.

1.  By Pearson\cite{Pearson1958}, in the absence of side walls, i.e. when the region extends infinitely in horizontal directions, the minimum
in (\ref{Mac}) is achieved at $\alpha_c\approx 2$ and the corresponding critical Marangoni number is $\lambda _{c}\approx 79.6$  when $\Bi=0$. 
And $\alpha _{c}$ 
increases monotonuously to approximately 3 as $\Bi
\rightarrow \infty $.

When the side walls are
present the minimum is achieved on the lattice given in (\ref{Mac}). Hence
the side walls have a stabilizing effect on the system by increasing the critical Marangoni number.

2. Let us consider first the case where a single mode becomes unstable at the critical Marangoni number. As an example, we will take the length scales to be $L_1=1.5$ and $L_2=1.0$. In this case, the critical index can be found to be $(I_x,I_y)=(1,0)$, the critical Marangoni number is $\lambda_c\approx79.82$ and the critical wave number is $\alpha_I=2\pi/3$ at $Bi=0$. This type of transition, according to Theorem \ref{single.thm} is controlled by the cubic term $\sigma$ in \eqref{1inter}. For various choices of Prandtl numbers, the value of $\sigma$ is shown in Figure \ref{1interactioncoef}. In this case we see that the transition at $\lambda_c$ is of Type-I (continuous). Moreover, by \eqref{1interstructure} the two bifurcated solutions will be a small perturbation of the critical mode $\phi_I$, one having a roll structure as shown in figure \ref{1interactionstructure} and the other rotating in the opposite direction.
\begin{figure}[tbp]
\includegraphics[scale=0.5]{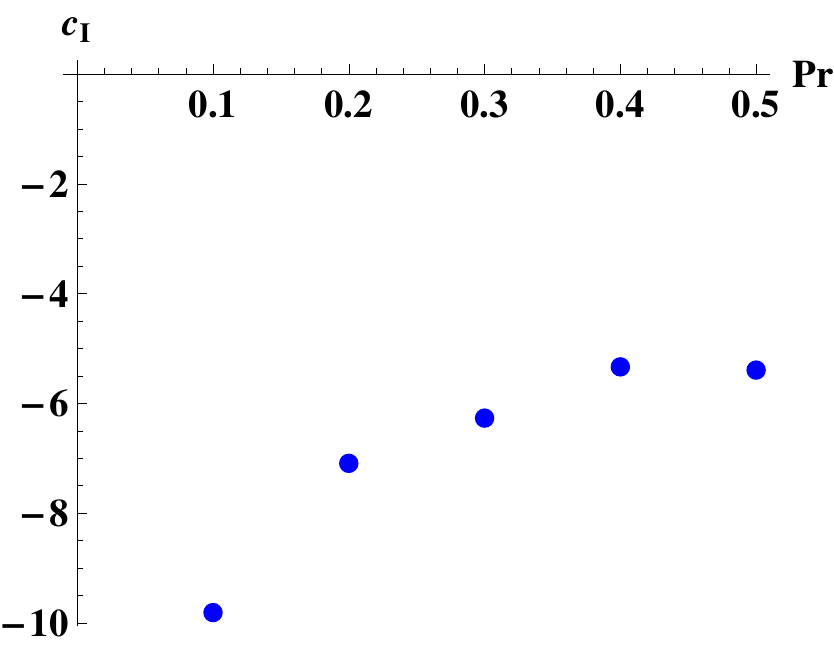}
\includegraphics[scale=0.5]{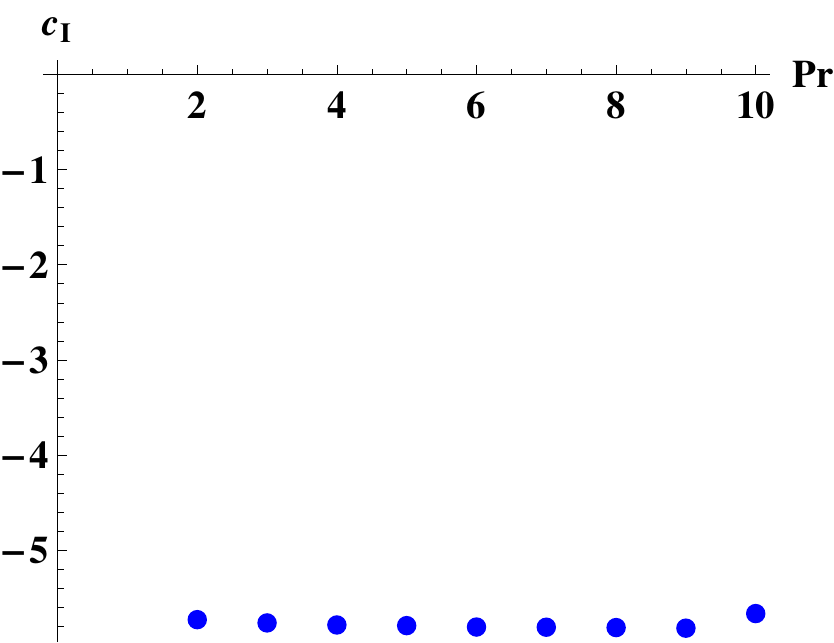}
\label{1interactioncoef}
\caption{Coefficient of the cubic term in \eqref{1inter} for various  $\Pr$ at $L_1=1.5$ and $L_2=1$ and $\Bi=0$.}
\end{figure}

\begin{figure}[tbp]
\includegraphics[scale=0.5]{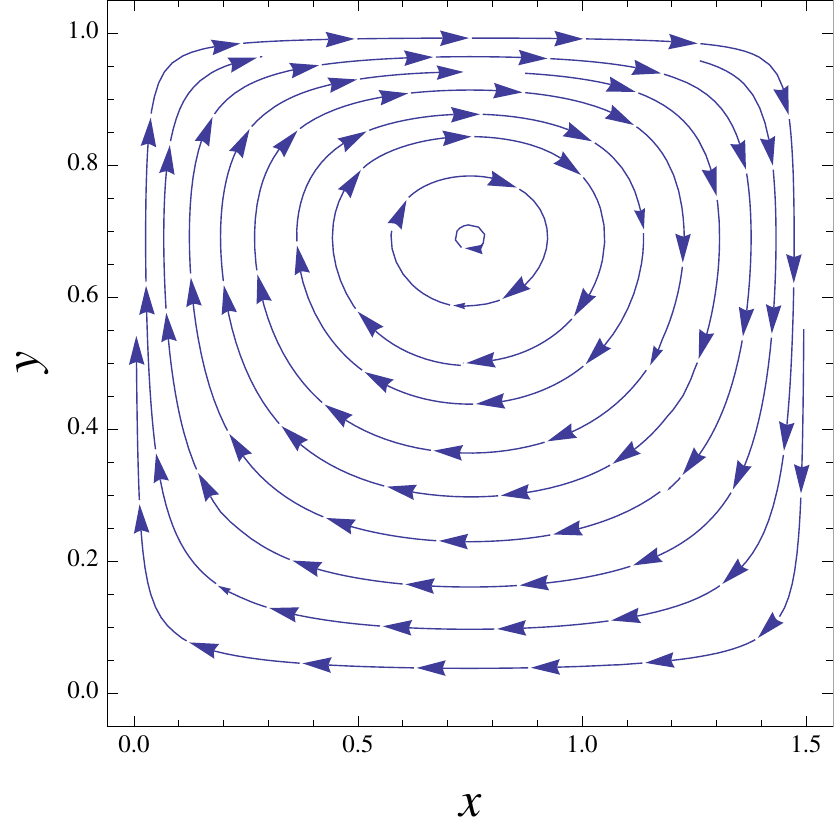}
\includegraphics[scale=0.4]{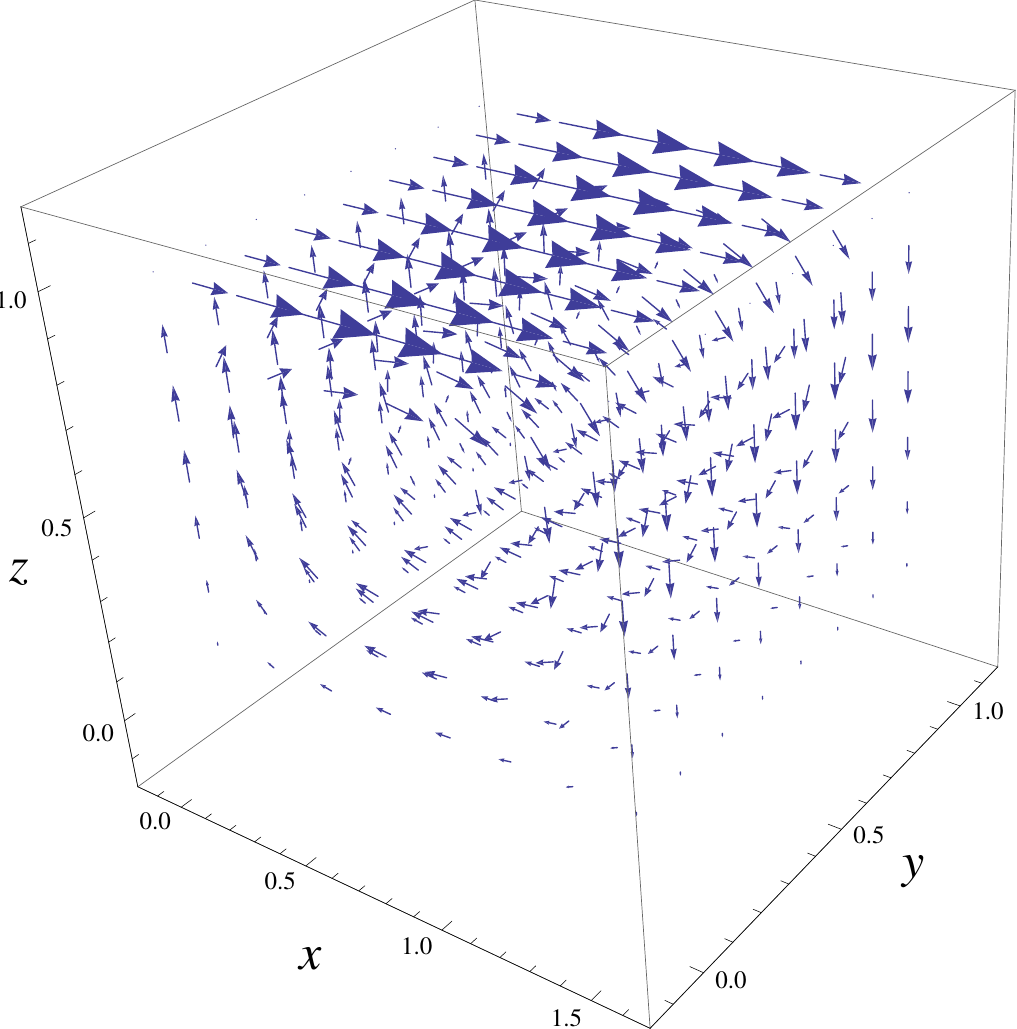}
\caption{The transition state has a time independent roll structure when $L_1=1.5$ and $L_2=1$ at Bi=0.}
\label{1interactionstructure}
\end{figure}

3. Consider a box with dimensions $L_2=3.02$ and $L_1=\frac{2L_2}{\sqrt{3}}$. Then the critical modes are found to be $I=(2,1)$ and $J=(0,2)$. At $Bi=0$, $\lambda_c\approx79.77$. The type of transition in this case depends on a number $b_2$ given by Theorem \ref{hex.thm}. We compute the values of $b_2$ for several Prandtl numbers $\Pr$ as shown in Figure \ref{b2val} which indicates a Type-III transition at $\lambda_c$. 
\begin{figure}[tbp]
\includegraphics[scale=0.5]{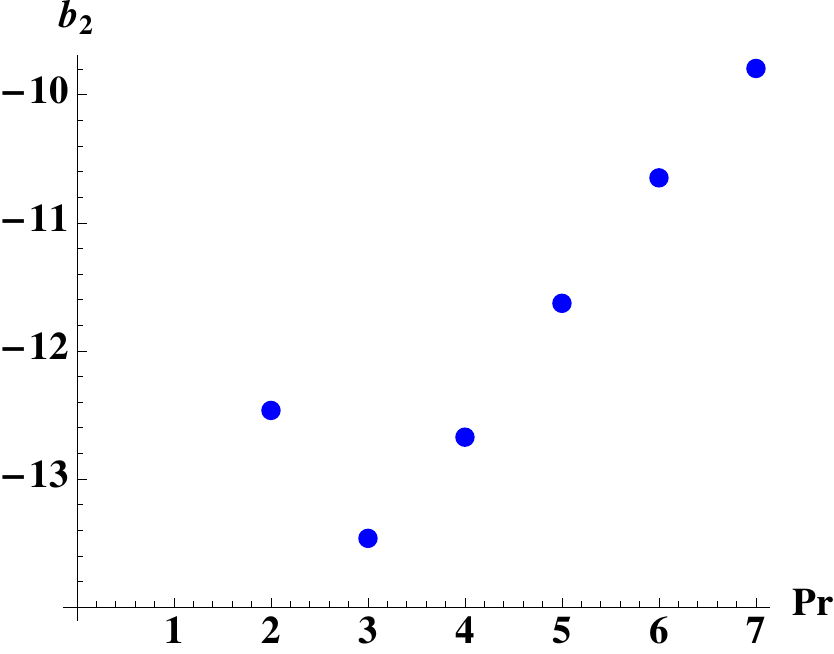}
\label{b2val}
\caption{The value of $b_2$ for various $Pr$ numbers for $L_2=3.02$ and $L_1=\frac{2L_2}{\sqrt{3}}$ and $Bi=0$. }
\end{figure}

\section{Proof of the Main Theorems}

By (\ref{1inter}),  the proof of Theorem~\ref{single.thm}   follows immediately from the standard dynamic transition theorem 
from the simple eigenvalue \cite{ptd}. We now prove Theorem \ref{hex.thm}. 
Under the assumptions, the set of quadratic interactions  becomes
$$\mathcal{S}=\{0,K,J,2I,2J,I+J,I\}$$
with $K=(2I_x,0)$ and $0=(0,0)$; see \eqref{S}. We can easily see that:
\begin{equation*}
\begin{aligned}
& G((I,i)),(I,j),(S,s))=0 && \text{if } S\neq 0,J,2I,K,\\
& G((I,i),(J,j),(S,s))=0 && \text{if } S\neq I+J,I, \\
& G((J,i),(J,j),(S,s))=0 && \text{if } S\neq 0,2J.
\end{aligned}
\end{equation*}
Thus  the center manifold function is given by 
\begin{align}
\Phi(y)=
&y_I^2 \sum_{\substack{l\geq1 \text{ and } S=0,2I,K \\l\geq2 \text{ and } S=J}} \Phi_{IIS}^l \phi_{(S,l)}\\
&+y_I y_J\sum_{\substack{l\geq2 \text{ and }S=I\\l\geq1 \text{ and }S=I+J}}\Phi_{IJS}^l\phi_{(S,l)}+y_J^2\sum_{l\geq1 \text{ and } S=0,2J} \Phi_{JJS}^l \phi_{(S,l)}, \nonumber 
\end{align}
and the reduced equation \eqref{reduced2} on the center manifold becomes, 
\begin{equation}\label{hex1}
\begin{aligned}
&
\frac{dy_I}{dt}=\beta (\lambda) y_I+a_1 y_J y_I + y_I(a_2 y_I^2+a_3 y_J^2)+o(3), \\
&
\frac{dy_J}{dt}=\beta (\lambda) y_J+b_1 y_I^2 + y_J(b_2 y_J^2+ b_3 y_I^2 )+o(3).
\end{aligned}
\end{equation}
The coefficients in \eqref{hex1} are as follows:
\begin{equation}\label{coef}
\begin{aligned}
& \beta=\beta_{(I,1)}=\beta_{(J,1)},\\
& a_2=\sum_{\substack{l\geq1 \text{ and } S=0,2I,K \\l\geq2 \text{ and } S=J}} H((I,1),(S,l),(I,1)) \Phi_{IIS}^l, \\
& a_3=\sum_{l\geq 1} G((I,1),(0,l),(I,1))\Phi_{JJ0}^l\\
    & \qquad +\sum_{\substack{l\geq2 \text{ and }S=I\\l\geq1 \text{ and }S=I+J}}H((J,1),(i,l),(I,1))\Phi_{IJS}^l, \\
& b_1=G((I,1),(I,1),(J,1)),  \\
& b_3=\sum_{l\geq1}G((J,1),(0,l),(J,1))\Phi_{II0}^l\\
& \qquad +\sum_{\substack{l\geq2 \text{ and }S=I\\l\geq1 \text{ and }S=I+J}}H((I,1),(i,l),(J,1)) \Phi_{IJS}^l.\\
\end{aligned}
\end{equation}
The coefficients $a_1$ and $b_2$ are given by \eqref{2p}.
The following relations can be found between the coefficients of \eqref{hex1}:
\begin{equation}\label{hexrel}
a_1=4b_1, \qquad a_3=2b_3, \qquad 4a_2=a_3+b_2.
\end{equation}
The reduced equations \eqref{hex1} are similar to those found by Dauby et al. \cite{Dauby1993}.
It is known that the type of transition is determined by the equation \eqref{hex1} at $\lambda=\lambda_c$. 
For this purpose, we look for the straight line orbits of \eqref{hex1} at $\lambda=\lambda_c$. Setting $y_I=k y_J$ at $\lambda=\lambda_c$ in \eqref{hex1} we find:
$$
k=\frac{dy_I}{dy_J}=\frac{a_1 k y_J^2+ y_J^3(a_2 k^3+a_3 k)}{b_1 k^2 y_J^2+y_J^3(b_2+k^2 b_3)}.
$$
Using the relations \eqref{hexrel} we find that there are six straight line orbits on the lines $k=0, \pm 2$. Also the bifurcated steady state solutions are given by \eqref{steady}. 

Now let us denote the Jacobian matrix of the vector field \eqref{hex1} by $J$. Then we find the eigenvalues of $J$ at $\pm R^{\lambda}$ to be 
\begin{equation}\label{pmR}
\left\{\pm a_1\sqrt{\frac{-\beta}{b_2}} +\beta(1-\frac{a_3}{b_2})+o(\beta),\,-2\beta+o(\beta)\right\}
\end{equation}
and the eigenvalues of $J$ at $(H_i^{\lambda})$ to be $\{2\beta+O(\beta^2),\,-\beta+O(\beta^2)\}$, $i=1,2$. This analysis shows that all the singular points are non-degenerate and moreover $H_i^{\lambda}$ are always saddle points.
As is well-known the type of transition is dictated by the reduced equation \eqref{hex1} at $\lambda=\lambda_c$.

\begin{itemize}
\item[I.]
If $b_2<0$, then at $\lambda=\lambda_c$, the origin has a parabolic region and four hyperbolic regions as shown in Figure~\ref{b2<0cr}. Thus the transition is Type-III at $\lambda=\lambda_c$ . $\pm R^{\lambda}$  are two steady states bifurcated from $(0,0)$ on $\lambda>\lambda_c$ and have the structure of rolls. By \eqref{pmR}, we find that $-R^{\lambda}$ is an attractor and $+R^{\lambda}$ is a saddle for $\beta>0$ small. The topological structure is as shown in Figure \ref{b2less0}. In the case $a_1<0$ we get the symmetric picture with respect to $y_I$ axis.

\item[II.]
 On the other hand we find that if $b_2>0$ then all the regions are hyperbolic at $\lambda=\lambda_c$ and thus the transition is of Type-II. 
\end{itemize}

\section{Well-Posedeness}

We will assume that the readers are familiar with the usual Sobolev spaces $H^{m}\left( \Omega \right) $ and the $L^{p}\left( \Omega \right)$ spaces.
For convenience, we will use $\left\vert \cdot \right\vert $ and $\left(\cdot ,\cdot \right) $ to indicate the standard norm and inner product in $%
L^{2}\left( \Omega \right) $.  The upper boundary ($z=1$) will be denoted by $\Gamma $.

\subsection{3D Case}

For the functional setting of the problem we follow Foias-Manley-Temam \cite{Foias1987}. 
First we define the relevant function spaces:
\begin{equation*}
\begin{aligned} 
& V_{1} =\left\{ \mathbf{u}\in H^{1}\left( \Omega \right)
^{3}:\nabla\cdot\mathbf{u}=0\text{, }\mathbf{u}\cdot n=0,\mathbf{u}\mid
_{z=0}=0\right\} , \\ 
& V_{2} =\left\{ \theta \in H^{1}\left( \Omega \right)
:\theta \mid _{z=0}=0\right\} ,\\ 
& H_{1} =\left\{ \mathbf{u}\in L^{2}\left(
\Omega \right) ^{3}:\nabla\cdot\mathbf{u}=0,\mathbf{u}\cdot n=0\right\} , H_{2} =L^{2}\left( \Omega \right),\\  
& V=V_{1}\times V_{2} \qquad H=H_{1}\times H_{2}, 
\end{aligned}
\end{equation*}
For $\phi =\left( \mathbf{u},\theta \right) \in V$, the linear operator $%
L_{\lambda }:V\rightarrow H$ is defined by 
\begin{equation}
\left( L_{\lambda }\phi ,\widetilde{\phi }\right) =-\text{Pr}\left( \nabla 
\mathbf{u},\nabla \widetilde{\mathbf{u}}\right) -\left( \nabla \theta
,\nabla \widetilde{\theta }\right) +\left( w,\widetilde{\theta }\right)
-\int_{\Gamma }\left( \text{Bi}\theta \widetilde{\theta }+\lambda \text{Pr}%
\nabla \theta \cdot \widetilde{\mathbf{u}}\right) .  \label{linearop}
\end{equation}%
Since the boundary integral term in (\ref{linearop}) contains $\nabla \theta 
$, it is not clear whether $L_{\lambda }$ is well defined on $V$ since $%
\nabla \theta \notin L^{2}\left( \Gamma \right) $ for general $\theta \in
V_{2}$. However this is true as it is shown in (\ref{trick}) that 
\begin{equation*}
\int_{\Gamma }\nabla \theta \cdot \widetilde{\mathbf{u}}=\int_{\Omega
}\nabla \theta \cdot \partial _{z}\widetilde{\mathbf{u}}.
\end{equation*}%
The bilinear operator $G:V\times V\rightarrow H$ is defined by%
\begin{equation}  \label{bilinearop}
G\left( \phi ,\widetilde{\phi }\right) =-\left( \mathcal{P}\left( \mathbf{u}%
\cdot \nabla \right) \widetilde{\mathbf{u}},\left( \mathbf{u}\cdot \nabla
\right) \widetilde{\theta }\right) .
\end{equation}%
where $\mathcal{P}$ denotes the Leray projector onto the divergence free
vector fields. The nonlinear operator will be denoted by, 
\begin{equation*}
G\left( \phi \right) =G\left( \phi ,\phi \right) .
\end{equation*}

Now we can write the problem in the abstract form as,%
\begin{equation}  \label{abstract}
\begin{aligned} & \frac{d\phi }{dt} =L_{\lambda }\phi +G\left( \phi \right)
\text{,\thinspace \thinspace \thinspace\ }\phi \in V, \\ & \phi \left(
0\right) =\phi _{0}. \end{aligned}
\end{equation}
First we will need to ensure the well-posedeness of the problem (\ref%
{abstract}). Namely,

\begin{problem}
\label{weakex} For fixed $\tau>0$ and given $\phi _{0}\in H$, find 
\begin{equation*}
\phi \in L^{2}\left( 0,\tau ;V\right) \cap L^\infty\left( 0,\tau ;H\right), \quad \phi_t \in L^2(0, \tau; (V \cap H^2)').
\end{equation*}%
satisfying
\begin{equation}  \label{weakform}
\begin{aligned} 
& \left(\frac{\partial \mathbf{u}}{\partial t}
,\widetilde{\mathbf{u}}\right) +\int \left( \mathbf{u}\cdot \nabla
\right) \mathbf{u}\cdot \widetilde{\mathbf{u}}+\Pr\left( \nabla
\mathbf{u},\nabla \widetilde{\mathbf{u}}\right) =-\Pr\,\lambda
\int_{\Gamma }\nabla \theta \cdot \widetilde{\mathbf{u}}, \\ 
& \left(\frac{\partial \theta}{\partial t}  ,\widetilde{\theta }\right) +\int
\left( \mathbf{u}\cdot \nabla \right) \theta \,\widetilde{\theta }+\left(
\nabla \theta ,\nabla \widetilde{\theta }\right) =\left( w,\theta \right)
-\text{Bi}\int_{\Gamma }\theta \widetilde{\theta }, 
\end{aligned}
\end{equation}
for every $\left( \widetilde{\mathbf{u}},\widetilde{\theta }\right) \in V$.
\end{problem}

\begin{theorem} \label{Wellposed1}
Given any $\phi_0 \in H$ there exists a solution for Problem \ref{weakex}.
\end{theorem}
\begin{proof}
The difficulty in obtaining the well-posedeness is that the boundary integral
in the first equation of \eqref{weakform} can not be controlled by the
remaining terms. Using the boundary conditions and the divergence
free condition,we can obtain the following:
\begin{equation}  \label{trick}
\begin{aligned} \int_{\Gamma }\nabla \theta \cdot \widetilde{\mathbf{u}}
&=\int_{\Gamma }\widetilde{u}\partial _{x}\theta +\widetilde{v}\partial
_{y}\theta =-\int_{\Gamma }\theta \left( \partial _{x}\widetilde{u}+\partial
_{y}\widetilde{v}\right) =\int_{\Gamma }\theta \partial _{z}\widetilde{w} \\
&=\int_{\partial \Omega }\theta \partial _{z}\widetilde{\mathbf{u}}\cdot
n=\int_{\Omega }\nabla \theta \cdot \partial _{z}\widetilde{\mathbf{u}}.
\end{aligned}
\end{equation}
Then formally putting $\mathbf{u}=\widetilde{\mathbf{u}}$ and $\theta =%
\widetilde{\theta }$ in (\ref{weakform}) and multiplying the second equation
in (\ref{weakform}) by $\gamma $ which has to be chosen properly, using (%
\ref{trick}) and thanks to the fact that nonlinear terms vanish, we get
\begin{equation} \label{3d.1}
\begin{split}
\frac{1}{2}\frac{d}{dt}\left( \text{Pr}^{-1}\left\vert \mathbf{u}\right\vert
^{2}+\gamma \left\vert \theta \right\vert ^{2}\right) &+\left( \left\vert
\nabla \mathbf{u}\right\vert ^{2}+\gamma \left\vert \nabla \theta
\right\vert ^{2}\right) +\gamma\, \text{Bi}\int_{\Gamma }\theta ^{2}\\ &=\gamma
\left( w,\theta \right) -\lambda \left( \nabla \theta ,\partial _{z}\mathbf{u%
}\right) .
\end{split}
\end{equation}
Choosing $\gamma>\max\{\lambda^2,\Pr^{-1}\}$, we can estimate the terms in the right handside of \eqref{3d.1} as
\begin{equation*}
\begin{aligned}
& \lambda |(\nabla \theta ,\partial _{z}\mathbf{u})| \leq \frac{1}{2}(|\nabla \mathbf{u}|^2+\lambda^2 |\nabla \theta|^2) \leq  \frac{1}{2}(|\nabla \mathbf{u}|^2+\gamma |\nabla \theta|^2) \\
& \gamma |(w,\theta)| \leq  \frac{1}{2}(\Pr^{-1} |\mathbf{u}|^2+\gamma^2 \Pr |\theta|^2) \leq \frac{ \gamma\, \Pr}{2} (\Pr^{-1} |\mathbf{u}|^2+ \gamma |\theta|^2)
\end{aligned}
\end{equation*}
So \eqref{3d.1} becomes
\begin{equation}  \label{3d.2}
\frac{d}{dt} ( \Pr^{-1}|\mathbf{u}|^{2}+\gamma | \theta| ^{2}) +(|
\nabla \mathbf{u}|^{2}+\gamma |\nabla \theta|^{2}) \leq c \, (\Pr^{-1} |\mathbf{u}|^2+ \gamma |\theta|^2)
\end{equation}
where $c=\gamma \,\Pr$. By Gronwall's Lemma, the above inequality gives
\begin{equation} \label{3d.est.1}
 \sup_{0\leq t \leq \tau} (\Pr^{-1}|\mathbf{u}|^{2}+\gamma | \theta| ^{2}) \leq \exp(c\, \tau)( \Pr^{-1}|\mathbf{u}_0|^{2}+\gamma | \theta_0| ^{2}) \\
\end{equation}
Integrating \eqref{3d.2} in time and using the above estimate,
\begin{equation} \label{3d.est.2}
 \int_0^\tau (\Pr^{-1}|\nabla \mathbf{u}|^{2}+\gamma |\nabla \theta|^{2})dt \leq c\, \tau \exp(c \, \tau)( \Pr^{-1}|\mathbf{u}_0|^{2}+\gamma | \theta_0| ^{2})
 \end{equation}
From the estimates \eqref{3d.est.1} and \eqref{3d.est.2}, the assertion of the theorem follows by standard Galerkin approximation as in the discussions of Navier-Stokes equations in Temam \cite{temam2001} .
\end{proof}
\subsection{2D Case}

Here we will prove the existence of the strong solutions in 2D.
Also we will prove that in 2D the system possesses a global attractor, a
compact set in $H$ invariant under the flow which attracts every bounded set
in $H$.

For the 2-dimensional case an equivalent formulation of the problem (\ref%
{main})-(\ref{bc}) can be given using the stream function formulation. The
reason we prefer this formulation is that we can obtain higher order energy
estimates without boundary terms in this case, see estimate (\ref{eq1}). In
this formulation we have:%
\begin{equation}  \label{streamfunc}
\begin{aligned} & \partial _{t}\Delta \psi +\left\{ \psi ,\Delta \psi
\right\} =\text{Pr}\Delta ^{2}\psi , \\ & \partial _{t}\theta +\left\{ \psi
,\theta \right\} =\partial _{x}\psi +\Delta \theta , \\ & \psi \left(
0\right) =\psi _{0}\text{, \medskip} \theta \left( 0\right) =\theta _{0}.
\end{aligned}
\end{equation}%
The domain is $\Omega =$ $\left\{ \left( x,z\right) \in \left( 0,L\right)
\times \left( 0,1\right) \right\} \subset \mathbb{R}^{2}$. Here, $\psi $ is
the stream function and $\theta $ is the temperature representing a small
perturbation from the basic state. The Jacobian is defined as%
\begin{equation*}
\left\{ f,g\right\} =\partial _{x}f\,\partial _{z}g-\partial _{z}f\,\partial
_{x}g.
\end{equation*}
For simplicity we assume slightly different boundary conditions than the 3D
case. Namely, we set $\text{Bi}=0$ at the upper boundary and we choose
free-slip boundary for the velocity at the bottom: 
\begin{equation}  \label{streamfuncBC}
\begin{aligned} & \psi =\Delta \psi =\partial _{x}\theta =0 &&\text{ at
}x=0,L. \\ & \psi =\Delta \psi =\theta =0 &&\text{ at }z=0. \\ & \psi
=\partial _{z}\theta =\Delta \psi -\lambda \partial _{x}\theta =0 &&\text{
at }z=1. \end{aligned}
\end{equation}

As before let us denote the upper boundary by $\Gamma $. Notice that if (\ref%
{streamfuncBC}) is satisfied for smooth $\left( \psi ,\theta \right) $ then $%
\partial_{n} \Delta \psi \mid _{\Gamma }=0$. Thus we define the following
function spaces: $H_{1}=H^{2}\left( \Omega \right) $, $H_{2}=L^{2}\left(
\Omega \right) $, $H=H_{1}\times H_{2}$, $V=V_{1}\times V_{2},$%
\begin{align*}
& V_{1} =\left\{ \psi \in H^{3}\left( \Omega \right) :\psi \mid _{\partial
\Omega }=0,\Delta \psi \mid _{\partial \Omega \backslash \Gamma }=0\right\} ,
\\
& V_{2} =\left\{ \theta \in H^{1}\left( \Omega \right) : \theta \mid
_{z=0}=0\right\} , \\
& D_{1} =\left\{ \psi \in H^{4}\left( \Omega \right) :\psi \mid _{\partial
\Omega }=0,\Delta \psi \mid _{\partial \Omega \backslash \Gamma }=0,\frac{%
\partial \Delta \psi }{\partial n}\mid _{\Gamma }=0\right\} , \\
& D_{2} =\left\{ \theta \in H^{2}\left( \Omega \right) :\theta \mid _{z=0}=0,%
\frac{\partial \theta }{\partial n}\mid _{\partial \Omega \backslash \left\{
z=0\right\} }=0\right\} .
\end{align*}

Integrating by parts, we can easily verify that for $\left( \psi ,\theta
\right) \in V$%
\begin{align}
& \left( \left\{ \psi ,\Delta \psi \right\} ,\Delta \psi \right) =0,
\label{id3} \\
& \left( \left\{ \psi ,\theta \right\} ,\theta \right) =0.  \label{id5}
\end{align}

Throughout the rest of the paper, by $c$ we denote a generic positive
constant depending possibly on $\text{Pr}$ and $\Omega $. By the elliptic
theory for the Laplacian operator (see Grisvard\cite{Grisvard1985}), for all 
$\psi \in V_{1}\subset H^{2}\left( \Omega \right) \cap H_{0}^{1}\left(
\Omega \right) $ we have 
\begin{equation}
\left\Vert \psi \right\Vert _{H^{2}\left( \Omega \right) }\leq c\left\vert
\Delta \psi \right\vert \text{,}  \label{equivnorm1}
\end{equation}%
and for all $\theta \in V_{2}$%
\begin{equation}
\left\Vert \theta \right\Vert _{V_{2}}\leq c\left( \left\vert \Delta \theta
\right\vert +\left\vert \theta \right\vert \right) ,  \label{equivnorm3}
\end{equation}%
and for all $\psi \in D_{1}$%
\begin{equation}
\left\Vert \Delta \psi \right\Vert _{H^{2}\left( \Omega \right) }\leq
c\left( \left\vert \Delta ^{2}\psi \right\vert +\left\vert \Delta \psi
\right\vert \right) .  \label{equivnorm4}
\end{equation}%
The boundary conditions allow us to use the following Poincare inequalities
for all $\psi \in V_{1}$ and $\theta \in V_{2}$%
\begin{equation}
\left\vert \varphi \right\vert \leq c\left\vert \nabla \varphi \right\vert 
\text{,\thinspace \thinspace \thinspace }\varphi =\psi \text{ or }\varphi
=\Delta \psi \text{ or }\varphi =\theta .  \label{poincare}
\end{equation}%
Note that(\ref{equivnorm1}) and (\ref{poincare}) implies that $\left\vert
\nabla \Delta \cdot \right\vert $ is an equivalent norm on $V_{1}$. In space
dimension two, we have, 

\begin{equation*}
\left\Vert \varphi \right\Vert _{L^{3}\left( \Omega \right) }\leq
c\left\Vert \varphi \right\Vert _{H^{1/2}\left( \Omega \right) }\leq
c\left\vert \varphi \right\vert ^{1/2}\left\Vert \varphi \right\Vert
_{H^{1}\left( \Omega \right) }^{1/2}\text{, for all }\varphi \in H^{1}\left(
\Omega \right)
\end{equation*}
Also we will make use of the interpolation theorem due to
Gagliardo-Nirenberg which is valid for space dimension two (see
Milani--Koksch\cite{Milani2005}):
\begin{equation}
\left\Vert \varphi \right\Vert _{L^{4}}\leq c\left\vert \varphi \right\vert
^{1/2}\left\vert \nabla \varphi \right\vert ^{1/2}+\widetilde{c}\left\vert
\varphi \right\vert \text{, }\forall \varphi \in H^{1}\left( \Omega \right) ,
\label{gagliardo}
\end{equation}%
where we can choose $\widetilde{c}=0$ if $\varphi \mid _{\partial \Omega }=0$%
.

\begin{theorem}
\label{Existence} For any fixed $\tau>0$, if $\left( \psi _{0},\theta
_{0}\right) \in H$ then there exists a unique solution 
\begin{equation}
\psi \in L^{2}\left( 0,\tau;V_{1}\right) \cap C\left( 0,\tau;H_{1}\right) ,
\label{weak1}
\end{equation}
\begin{equation}
\theta \in L^{2}\left( 0,\tau;V_{2}\right) \cap C\left( 0,\tau;H_{2}\right) ,
\label{weak3}
\end{equation}%
satisfying 
\begin{align}
& \frac{d}{dt}\left( \nabla \psi ,\nabla \widetilde{\psi }\right) =-\text{Pr}
\left( \Delta \psi ,\Delta \widetilde{\psi }\right) +\left( \left\{ 
\widetilde{\psi },\psi \right\} ,\Delta \psi \right) +\text{Pr} \lambda
\int_{\Gamma }\partial _{x}\theta \frac{\partial \widetilde{\psi }}{\partial
n},  \label{weak2} \\
& \frac{d}{dt}\left( \theta ,\widetilde{\theta }\right) =-\left( \nabla
\theta ,\nabla \widetilde{\theta }\right) -\left( \left\{ \psi ,\theta
\right\} ,\widetilde{\theta }\right) +\left( \partial _{x}\psi ,\widetilde{%
\theta }\right) ,  \label{weak4}
\end{align}%
for all $\left( \widetilde{\psi },\widetilde{\theta }\right) \in V$.
Moreover, if $\left( \psi _{0},\theta _{0}\right) \in V$ then%
\begin{eqnarray*}
\psi &\in &L^{2}\left( 0,\tau;D_{1}\right) \cap C\left( 0,\tau;V_{1}\right) ,
\\
\theta &\in &L^{2}\left( 0,\tau;D_{2}\right) \cap C\left(
0,\tau;V_{2}\right) .
\end{eqnarray*}
\end{theorem}

\begin{proof}
Formally putting $\widetilde{\psi }=\Delta \psi $ in (\ref{weak2}) and using
(\ref{id3}) we get,%
\begin{equation}
\frac{1}{2}\frac{d}{dt}\left\vert \Delta \psi \right\vert ^{2}+\text{Pr}%
\left\vert \nabla \Delta \psi \right\vert ^{2}=0.  \label{eq1}
\end{equation}

Now we put $\theta $ in (\ref{weak4}) and use (\ref{id5}), 
(\ref{equivnorm1}) and (\ref{poincare}) to get
\begin{eqnarray*}
\frac{1}{2}\frac{d}{dt}\left\vert \theta \right\vert ^{2}+\left\vert \nabla
\theta \right\vert ^{2} &\leq &\left\vert \nabla \psi \right\vert \left\vert
\theta \right\vert \\
&\leq &c\left\vert \Delta \psi \right\vert \left\vert \nabla \theta
\right\vert \\
&\leq &c\left\vert \Delta \psi \right\vert ^{2}+\frac{1}{2}\left\vert \nabla
\theta \right\vert ^{2}
\end{eqnarray*}
which yields
\begin{equation}
\frac{d}{dt}\left\vert \theta \right\vert ^{2}+\left\vert \nabla \theta
\right\vert ^{2}\leq c\left\vert \Delta \psi \right\vert ^{2}.  \label{eq3}
\end{equation}

Using the standard Galerkin approximation, we can derive (\ref{weak1}) and (\ref{weak3}) 
from the estimates (\ref{eq1}) and (\ref{eq3}).
\end{proof}

By Theorem (\ref{Existence}), we can define a continuous semigroup $S\left(
t\right) :H\rightarrow H$, 
\begin{equation*}
S\left( t\right) \left( \psi _{0},\theta _{0}\right) \rightarrow \left( \psi
\left( t\right) ,\theta \left( t\right) \right) \text{,\thinspace \thinspace
\thinspace \thinspace \thinspace \thinspace \thinspace \thinspace \thinspace
\thinspace \thinspace \thinspace\ }t\geq 0.
\end{equation*}

\begin{theorem}
\label{attractor}There exist bounded absorbing sets $B_{H}\left( 0,\rho
_{1}\right) $ in $H$ and $B_{V}\left( 0,\rho _{2}\right) $ in $V$ for the
semigroup $S\left( t\right) $. Moreover $S\left( t\right) $ possesses a
unique global attractor $\mathcal{A}\subset B_{H}\left( 0,\rho _{1}\right) $
which is compact, connected and maximal in $H$ such that $\mathcal{A}$
attracts all the bounded subsets in $H$.
\end{theorem}

\begin{proof}
By (\ref{poincare}), we can write (\ref{eq1}) as,%
\begin{equation}
\frac{d}{dt}\left\vert \Delta \psi \right\vert ^{2}+c\left\vert \Delta \psi
\right\vert ^{2}\leq 0.  \label{eq2}
\end{equation}%
By Gronwall's inequality, (\ref{eq2}) implies%
\begin{equation}
\left\vert \Delta \psi \left( t\right) \right\vert ^{2}\leq \left\vert
\Delta \psi _{0}\right\vert ^{2}\exp \left( -ct\right) \text{, }\forall
t\geq 0.  \label{decaylappsi}
\end{equation}%
Using (\ref{poincare}), (\ref{decaylappsi}) and Gronwall's inequality, we
obtain from (\ref{eq3}) that%
\begin{equation}
\left\vert \theta \left( t\right) \right\vert ^{2}\leq c\left( \left\vert
\theta _{0}\right\vert ^{2}+\left\vert \Delta \psi _{0}\right\vert
^{2}\right) \exp \left( -ct\right) .  \label{decaytheta}
\end{equation}%
The estimates (\ref{decaylappsi}) and (\ref{decaytheta}) give immediately
the existence of an absorbing set $B_{H}\left( 0,\rho _{1}\right) $ in $H$.
In fact any $H$-ball around zero is an absorbing set in $H$. Now, we will
prove the existence of an absorbing set in $V$.

Multiplying the second equation in (\ref{streamfunc}) by $\Delta \theta $
and integrating over $\Omega $.%
\begin{eqnarray}
\frac{1}{2}\frac{d}{dt}\left\vert \nabla \theta \right\vert ^{2}+\left\vert
\Delta \theta \right\vert ^{2} &\leq &\left\vert \left( \partial _{x}\psi
,\Delta \theta \right) \right\vert +\left\vert \left( \left\{ \psi ,\theta
\right\} ,\Delta \theta \right) \right\vert  \notag \\
&\leq &\frac{1}{4}\left\vert \Delta \theta \right\vert ^{2}+\left\vert
\nabla \psi \right\vert ^{2}+\left\vert \left( \left\{ \psi ,\theta \right\}
,\Delta \theta \right) \right\vert .  \label{eq7}
\end{eqnarray}%
We deal with the nonlinear term in (\ref{eq7}) as follows:%
\begin{eqnarray}
\left\vert \left( \left\{ \psi ,\theta \right\} ,\Delta \theta \right)
\right\vert &\leq &\int \left\vert \nabla \psi \right\vert \left\vert \nabla
\theta \right\vert \left\vert \Delta \theta \right\vert  \notag \\
&\leq &\left\Vert \nabla \psi \right\Vert _{L^{4}}\left\Vert \nabla \theta
\right\Vert _{L^{4}}\left\vert \Delta \theta \right\vert  \notag \\
&\leq &c\left\vert \Delta \psi \right\vert \left\Vert \nabla \theta
\right\Vert _{L^{4}}\left\vert \Delta \theta \right\vert  \notag \\
&\leq &\text{(by (\ref{gagliardo}))}  \notag \\
&\leq &c\left( \left\vert \Delta \psi \right\vert \left\vert \nabla \theta
\right\vert ^{1/2}\left\vert \Delta \theta \right\vert ^{3/2}+\left\vert
\Delta \psi \right\vert \left\vert \nabla \theta \right\vert \left\vert
\Delta \theta \right\vert \right)  \notag \\
&\leq &c\left( \left\vert \Delta \psi \right\vert ^{4}+\left\vert \Delta
\psi \right\vert ^{2}\right) \left\vert \nabla \theta \right\vert ^{2}+\frac{%
1}{4}\left\vert \Delta \theta \right\vert ^{2}.  \label{eq4}
\end{eqnarray}%
Using (\ref{eq4}) in (\ref{eq7}), we get%
\begin{equation}
\frac{d}{dt}\left\vert \nabla \theta \right\vert ^{2}+\left\vert \Delta
\theta \right\vert ^{2}\leq 2\left\vert \nabla \psi \right\vert ^{2}+c\left(
\left\vert \Delta \psi \right\vert ^{4}+\left\vert \Delta \psi \right\vert
^{2}\right) \left\vert \nabla \theta \right\vert ^{2}.  \label{eq5}
\end{equation}%
By (\ref{decaylappsi}), from (\ref{eq5}) we can show that $\exists
t_{0}\left( \mathcal{B}\right) >0$ where $\mathcal{B}$ is any open ball in $%
V $ containing the given initial data s.t.%
\begin{equation}
\left\vert \nabla \theta \right\vert ^{2}\leq c\left( \left\vert \theta
_{0}\right\vert ^{2}+\left\vert \nabla \theta _{0}\right\vert
^{2}+\left\vert \Delta \psi _{0}\right\vert ^{2}\right) \exp \left(
-ct\right) \text{, }\forall t\geq t_{0}.  \label{res6}
\end{equation}

Multiplying the first equation in (\ref{streamfunc}) by $\Delta ^{2}\psi $
and integrating over $\Omega $ and using similar estimates as above we get%
\begin{eqnarray*}
\frac{1}{2}\frac{d}{dt}\left\vert \nabla \Delta \psi \right\vert ^{2}+\text{%
Pr}\left\vert \Delta ^{2}\psi \right\vert ^{2} &=&\left( \left\{ \psi
,\Delta \psi \right\} ,\Delta ^{2}\psi \right) \\
&\leq &c\left( \left\vert \Delta \psi \right\vert ^{3}\left\vert \nabla
\Delta \psi \right\vert +\left\vert \Delta \psi \right\vert ^{4}\left\vert
\nabla \Delta \psi \right\vert ^{2}\right) +\frac{\text{Pr}}{2}\left\vert
\Delta ^{2}\psi \right\vert ^{2}.
\end{eqnarray*}%
Thus we have%
\begin{equation}
\frac{d}{dt}\left\vert \nabla \Delta \psi \right\vert ^{2}+\text{Pr}%
\left\vert \Delta ^{2}\psi \right\vert \leq c\left( \left\vert \Delta \psi
\right\vert ^{4}+\left\vert \Delta \psi \right\vert ^{6}\right) \left\vert
\nabla \Delta \psi \right\vert ^{2}.  \label{eq9}
\end{equation}%
From (\ref{eq9}), it is standard to obtain the existence of some $%
t_{0}\left( \mathcal{B}\right) >0$ s.t.%
\begin{equation}
\left\vert \nabla \Delta \psi \left( t\right) \right\vert \leq C \quad \forall t\geq t_{0}.\label{res5}
\end{equation}
by means of the uniform Gronwall inequality, see Temam\cite{temam88}.
Finally, (\ref{res5}) and (\ref{res6}) imply the existence of an absorbing
ball in $V$. That proves the theorem.
\end{proof}

\section{Summary and discussion}
We have presented a rigorous analysis of dynamic transitions of the B\'enard-Marangoni
problem in a rectangular container. It adds to previous theoretical analyses through (i) the 
determination of the type of transitions, (ii) the attractors/repellors  associated with  
these transitions and (iii) the proof of the well-posedness of the mathematical problem.

When the container dimensions are such that the 
critical eigenvalue is simple,  Theorem  2 shows that the value of a computable scalar 
$c_I$ determines whether the transition is of Type-I ($c_I < 0$) or Type-II 
($c_I> 0$). In the example considered, we find a Type-I transition for all values of 
Pr which is in  agreement with the numerical computations in two-dimensional 
containers, indicating supercritical rolls for these cases \cite{Dijkstra1998}. 

When the container dimensions satisfy (5.4) to allow for hexagonal patterns to appear,
Theorem 3 shows that there is either a Type-II ($b_2 > 0$) or a Type-III ($b_2 < 0$) 
transition, again depending on a computable quantity ($b_2$).  In the example chosen 
here, we find $b_2 < 0$ for all values of Pr considered and hence a Type-III transition 
is guaranteed. This is consistent with experiments and numerical computations where 
hexagonal patterns are found below the value of the critical Marangoni number 
$\lambda_c$ \cite{Dauby1993, kosch}. 

Unfortunately, the theory cannot provide a statement on how the wavelength of the 
hexagonal patterns changes with increasing Marangoni number $\lambda$.   This is 
one of the  experimentally observed features  of these  patterns which, although also 
found in numerical models, still awaits a satisfactory explanation \cite{kosch}.

\appendix
\section{Linear Problem}
We now describe all the eigenpairs of the linear problem (\ref{Lin-0})  and (\ref{bc}) 
and the adjoint problem (\ref{adj1})  and (\ref{adjbc}). The solutions having $\theta\equiv 0$ 
will not play a role since those modes can be shown to be neither critical nor 
contributes to the center manifold approximation. So we will assume $\theta\neq0$. 

{\sc Case $\alpha =0$:}  By (\ref{sepofvar}) 
$u=v=0$ which also implies $w=0$ by the divergence free condition. Thus the
linear equations reduce to:%
\begin{equation}
\begin{aligned} & \left( D^{2}-\alpha ^{2}-\beta \right) \Theta =0,\\ &
\Theta \left( 0\right) =D\Theta \left( 1\right) +\text{Bi}\Theta \left(
1\right) =0. 
\end{aligned}
\label{zero.mode}
\end{equation}
whose solutions for $\rho >0$ are, 
\begin{equation*}
\begin{aligned}
& \theta_{(0,l)} =\Theta_{(0,l)}=\sin \rho z,\\ 
& \beta_{(0,l)}=-\rho_l^2. 
\end{aligned}
\end{equation*}%
Here $\rho _{l}$, $l=0,1,2,\dots $ are the positive solutions of the equation 
\begin{equation*}
\rho +\text{Bi}\,\tan \rho =0,
\end{equation*}%
When $\text{Bi}=0$, $\rho _{l}=\frac{\pi }{2}+l\pi $. For $\text{Bi}\neq 0$,
\begin{equation*}
\frac{\pi }{2}+l\pi \leq \rho _{l}\leq \pi +l\pi .
\end{equation*}
Since \eqref{zero.mode} is self-adjoint, we have
$$
\phi_{(0,l)}=\phi_{(0,l)}^{\ast}, \qquad \beta_{(0,l)}=\beta_{(0,l)}^{\ast}.
$$

\medskip

{\sc Case $\alpha_I= \alpha\neq0$:}  The eigenpairs are determined by the equations \eqref{Lin3}-\eqref{adjodebc}.
We have $w\neq
0 $, otherwise (\ref{Lin3}) and (\ref{linbc}) would contradict that $\theta
\not\equiv 0$. For $\beta=0$ we find,
\begin{equation}  \label{criticaleigen}
\begin{aligned} 
& W_{(I,i)} =4\alpha ^{2}\left( 1+Cz\right) \sinh \alpha z-4\alpha^{3}z\cosh \alpha z, \\ 
& \Theta_{(I,i)} =\left( \Theta _{1}+Cz+\alpha^{2}z^{2}\right) \sinh \alpha z-\alpha \left( 3z+Cz^{2}\right) \cosh \alpha z, \\ 
& \Theta_{(I,i)} ^{\ast } =\Theta _{1}^{\ast }\sinh \alpha z, \\
& W_{(I,i)}^{\ast } =\left( w_{1}^{\ast }+w_{2}^{\ast }z+w_{3}^{\ast
}z^{2}\right) \sinh \alpha z-\alpha w_{1}^{\ast }z\cosh \alpha z, \\
\end{aligned}
\end{equation}
with
\begin{equation*}
\begin{aligned} 
& \Theta _{1} =\frac{\left( 1+\Bi\right) \alpha \left( \cosh \alpha\sinh \alpha +\alpha \right) +\left( 1+\Bi+\alpha ^{2}\right) \sinh^{2}\alpha }{\left( \alpha \cosh \alpha +\Bi\,\sinh \alpha \right) \sinh
\alpha }, \\ 
& C =\alpha \coth \alpha -1. \\
& \Theta _{1}^{\ast } =8\Pr w_{3}^{\ast }, \\ & w_{1}^{\ast } =-\sinh \alpha
\left( \alpha \cosh \alpha \,+\sinh \alpha \right) , \\ 
& w_{2}^{\ast }=-\left( 2\alpha ^{2}\cosh ^{2}\alpha -\alpha \cosh \alpha \,\sinh \alpha
-\left( 1+\alpha ^{2}\right) \sinh ^{2}\alpha \right) , \\ 
& w_{3}^{\ast }=\alpha \left( \alpha -\cosh \alpha \,\sinh \alpha \right) . 
\end{aligned}
\end{equation*}

\medskip

{\sc Case $w\neq 0,\theta\neq 0,\alpha=\alpha_I\neq 0,\beta\neq 0$:} 
The eigenmodes are given by 
\begin{align*}
& W_{(I,i)}=(\Pr-1)\beta(-k\eta\sinh\alpha z-\cosh\alpha z +k\alpha \sinh\eta z+\cosh\eta z), \\
&  \Theta_{(I,i)}=(1-\Pr)[k\eta\sinh\alpha z+\cosh\alpha z]+\Pr[k\alpha\sinh\eta z+\cosh\eta z] \\
&\qquad +b\sinh \mu z-\cosh\mu z,  \\
& W^*_{(I,i)}=w_1\sinh\alpha z+w_2\sinh\bar{\mu} z+w_3(\cosh\bar{\eta} z-\cosh\alpha z)+w_4\sinh\bar{\eta} z, \\
& \Theta^*_{(I,i)}=\alpha^{-2}(\Pr-1)\bar{\beta}^2(\bar{\eta}\sinh\alpha\cosh\bar{\eta}-\alpha\cosh\alpha\sinh\bar{\eta}) \sinh\bar{\mu} z,
\end{align*}
where
\begin{align*}
& \eta=\eta_{(I,i)}=\sqrt{\alpha^2+\beta_{(I,i)} \Pr^{-1}}, \\
& \mu=\mu_{(I,i)}=\sqrt{\alpha^2+\beta_{(I,i)}}, \\
& k=\frac{\cosh\eta-\cosh\alpha}{\eta\sinh\alpha-\alpha\sinh\eta}, \\
& w_1=(\Pr-1)\bar{\eta}\cosh\bar{\eta}\sinh\bar{\mu}+\cosh\alpha(\bar{\mu}\sinh\bar{\eta}-\Pr\bar{\eta}\sinh\bar{\mu}),\\
& w_2=\bar{\eta}\cosh\bar{\eta}\sinh\alpha-\alpha\cosh\alpha\sinh\bar{\eta},\\
& w_3=(\Pr-1)\alpha\sinh\bar{\eta}\sinh\bar{\mu}+\sinh\alpha(\bar{\mu}\sinh\bar{\eta}-\Pr\bar{\eta}\sinh\bar{\mu}), \\
& w_4=\Pr\alpha\cosh\alpha\sinh\bar{\mu}-\cosh\bar{\eta}(\bar{\mu}\sinh\alpha+(\Pr-1)\alpha\sinh\bar{\mu}), 
\end{align*}
and $b$ is a constant determined by the condition $D\Theta_{(I,i)}+Bi\Theta_{(I,i)}=0$ at $z=1$. The eigenvalue $\beta_{(I,i)}$ in this case can be found by solving the relation:
\begin{equation}
\lambda=\frac{-D^2W_{(I,i)}(1)}{\alpha_I^2\Theta_{(I,i)}(1)}.
\label{neg.ev}
\end{equation}
Given $\lambda$, $\Bi$, $\Pr$ and $\alpha$, the relation \eqref{neg.ev} has to be solved numerically for $\beta_{(I,i)}$.

\bibliographystyle{amsplain}
\bibliography{marangoni}

\end{document}